\documentclass{amsart}%
\usepackage{amsfonts}
\usepackage{amsmath}
\usepackage{amssymb}
\usepackage{graphicx}%
\newtheorem{theorem}{Theorem}
\theoremstyle{plain}

\newtheorem{definition}{Definition}

\newtheorem{lemma}{Lemma}

\newtheorem{remark}{Remark}

\numberwithin{equation}{section}
\ifx\pdfoutput\relax\let\pdfoutput=\undefined\fi
\newcount\msipdfoutput
\ifx\pdfoutput\undefined\else
\ifcase\pdfoutput\else
\ifx\paperwidth\undefined\else
\ifdim\paperheight=0pt\relax\else\pdfpageheight\paperheight\fi
\ifdim\paperwidth=0pt\relax\else\pdfpagewidth\paperwidth\fi
\fi\fi\fi
\begin{document}
\title[$p$-adic Invariant Quantum Fields\ via \ White Noise \ Analysis]{Construction of $p$-adic Covariant Quantum Fields in the Framework of White
Noise Analysis }
\author{Edilberto Arroyo-Ortiz}
\author{W. A. Z\'{u}\~{n}iga-Galindo}
\address{Centro de Investigaci\'{o}n y de Estudios Avanzados del Instituto
Polit\'{e}cnico Nacional\\
Departamento de Matem\'{a}ticas, Unidad Quer\'{e}taro\\
Libramiento Norponiente \#2000, Fracc. Real de Juriquilla. Santiago de
Quer\'{e}taro, Qro. 76230\\
M\'{e}xico.}
\email{earroyo@math.cinvestav.mx, \ wazuniga@math.cinvestav.edu.mx}
\thanks{The second author was partially supported by Conacyt Grant No. 250845.}

\begin{abstract}
In this article we construct a large class of interacting Euclidean quantum
field theories, over a $p$-adic space time, by using white noise calculus. We
introduce $p$-adic versions of the Kondratiev and Hida spaces in order to use
the Wick calculus on the Kondratiev spaces. The quantum fields introduced here
fulfill all the Osterwalder-Schrader axioms, except the reflection positivity.

\end{abstract}
\maketitle

\section{Introduction}

In this article, we construct interacting Euclidean quantum field theories,
over a $p$-adic spacetime, in arbitrary dimension, which satisfy all the
Osterwal\-der-Schrader axioms \cite{Osterwalder-Schrader} except for
reflection positivity. More precisely, we present a $p$-adic analogue of the
interacting field theories constructed by Grothaus and Streit in
\cite{GS1999}. The basic objects of an Euclidean quantum field theory are
probability measures on distributions spaces, in the classical case, on the
space of tempered distributions $\mathcal{S}^{\prime}(\mathbb{R}^{n})$. In
conventional quantum field theory (QFT) there have been some studies devoted
to the optimal choice of the space of test functions. In \cite{Jaffe}, Jaffe
discussed this topic (see also \cite{Lopu} and \cite{Strocchi}); his
conclusion was that, rather than an optimal choice, there exists a set of
conditions that must be satisfied by the candidate space, and any class of
test functions with these properties should be considered as valid. The main
condition is that the space of test functions must be a nuclear countable
Hilbert one. This fact constitutes the main mathematical motivation the study
of QFT on general nuclear spaces.

A physical motivation for studying QFT in the $p-$adic setting comes from the
conjecture of Volovich stating that spacetime has a non-Archimedean nature at
the Planck scale, \cite{Vol}, see also \cite{Var1}. The existence of the
Planck scale implies that below it the very notion of measurement as well as
the idea of `infinitesimal length' become meaningless, and this fact
translates into the mathematical statement that the Archimedean axiom is no
longer valid, which in turn drives to consider models based on $p$-adic
numbers. In the \ $p$-adic framework, the relevance of constructing quantum
field theories was stressed in \cite{V-V-Z} and \cite{Var2}. In the last 35
years \ $p$-adic QFT has attracted a lot of attention of physicists and
mathematicians, see e.g. \cite{Abd et al}, \cite{Dra01}-\cite{EU66},
\cite{Gubser}, \cite{Koch-Sait}-\cite{Koch}, \cite{Kh1}-\cite{LM89},
\cite{Mendoza-Zuniga}-\cite{Mis2}, \cite{Smirnov}-\cite{Smirnov-2},
\cite{Var1}-\cite{Zuniga-LNM-2016}, and the references therein.

A $p$-adic number is a sequence of the form%
\begin{equation}
x=x_{-k}p^{-k}+x_{-k+1}p^{-k+1}+\ldots+x_{0}+x_{1}p+\ldots,\text{ with }%
x_{-k}\neq0\text{,} \label{p-adic-number}%
\end{equation}
where $p$ denotes a fixed prime number, and the $x_{j}$s \ are $p$-adic
digits, i.e. numbers in the set $\left\{  0,1,\ldots,p-1\right\}  $. There are
natural field operations, sum and multiplication, on series of form
(\ref{p-adic-number}). The set of all possible $p$-adic sequences constitutes
the field of $p$-adic numbers $\mathbb{Q}_{p}$. The field $\mathbb{Q}_{p}$ can
not be ordered. There is also a natural norm in $\mathbb{Q}_{p}$ defined as
$\left\vert x\right\vert _{p}=p^{k}$, for a nonzero $p$-adic number $x$ of the
form (\ref{p-adic-number}). The field of $p$-adic numbers with the distance
induced by $\left\vert \cdot\right\vert _{p}$ is a complete ultrametric space.
The ultrametric property refers to the fact that $\left\vert x-y\right\vert
_{p}\leq\max\left\{  \left\vert x-z\right\vert _{p},\left\vert z-y\right\vert
_{p}\right\}  $ for any $x$, $y$, $z$ in $\mathbb{Q}_{p}$. As a topological
space, $\left(  \mathbb{Q}_{p},\left\vert \cdot\right\vert _{p}\right)  $ is
completely disconnected, i.e. the connected components are points. The field
of $p$-adic numbers has a fractal structure, see e.g. \cite{A-K-S},
\cite{V-V-Z}. All these results can be extended easily to $\mathbb{Q}_{p}^{N}%
$, see Section \ref{Sction_2}.

In \cite{Zuniga-FAA-2017}, see also \cite[Chapter 11]{KKZuniga}, the second
author introduced a class of non-Archimedean massive Euclidean fields, in
arbitrary dimension, which are constructed as solutions of certain covariant
$p$-adic stochastic pseudodifferential equations, by using techniques of white
noise calculus. In particular a new non-Archimedean Gel'fand triple was
introduced. By using this new triple, here we introduce non-Archimedean
versions of the Kondratiev and Hida spaces, see Section \ref{Section_4}. The
non-Archimedean Kondratiev spaces, denoted as $\left(  \mathcal{H}_{\infty
}\right)  ^{1}$, $\left(  \mathcal{H}_{\infty}\right)  ^{-1}$, play a central
role in this article.

Formally an interacting field theory with interaction $V$ \ has associated a
measure of the form%
\begin{equation}
d\mu_{V}=\frac{\exp\left(  -\int_{\mathbb{Q}_{p}^{N}}V\left(  \boldsymbol{\Phi
}\left(  x\right)  \right)  d^{N}x\right)  d\mu}{\int\exp\left(
-\int_{\mathbb{Q}_{p}^{N}}V\left(  \boldsymbol{\Phi}\left(  x\right)  \right)
d^{N}x\right)  d\mu}, \label{Eq_A1}%
\end{equation}
where $\mu$ is the Gaussian white noise measure, $\boldsymbol{\Phi}\left(
x\right)  $ is a random process at the point $x\in\mathbb{Q}_{p}^{N}$. \ In
general $\boldsymbol{\Phi}\left(  x\right)  $ is not an integrable function
rather a distribution, thus a natural problem is how to define $V\left(
\boldsymbol{\Phi}\left(  x\right)  \right)  $. For a review about the
techniques for regularizing $V\left(  \boldsymbol{\Phi}\left(  x\right)
\right)  $ and the construction of the associated measures, the reader may
consult \cite{Glimm-Jaffe}, \cite{GS1999}, \cite{Simon}, \cite{Strocchi} and
the references therein.

Following \cite{GS1999}, we consider the following generalized white
functional:%
\begin{equation}
\boldsymbol{\Phi}_{H}=\exp^{\lozenge}\left(  -\int_{\mathbb{Q}_{p}^{N}%
}H^{\lozenge}\left(  \boldsymbol{\Phi}\left(  x\right)  \right)
d^{N}x\right)  , \label{Eq_A2}%
\end{equation}
where $H$ is analytic function at the origin satisfying $H(0)=0$. The Wick
ana\-lytic function $H^{\lozenge}\left(  \boldsymbol{\Phi}\left(  x\right)
\right)  $ of process $\boldsymbol{\Phi}\left(  x\right)  $ coincides with the
usual Wick ordered function $:H\left(  \boldsymbol{\Phi}\left(  x\right)
\right)  :$ when $H$ is a polynomial function. It turns out that $H^{\lozenge
}\left(  \boldsymbol{\Phi}\left(  x\right)  \right)  $ is a distribution from
the Kondratiev space $\left(  \mathcal{H}_{\infty}\right)  ^{-1}$, \ and
consequently, its integral belong to $\left(  \mathcal{H}_{\infty}\right)
^{-1}$, if it exists. In general we cannot take the exponential of \ $-\int
H^{\lozenge}\left(  \boldsymbol{\Phi}\left(  x\right)  \right)  d^{N}x$,
however, by using the Wick calculus in $\left(  \mathcal{H}_{\infty}\right)
^{-1}$, see Section \ref{Section Wick_analytic_fun}, we can take the Wick
exponential $\exp^{\lozenge}\left(  \cdot\right)  $.

In certain cases, for instance when $H$ is linear or is a polynomial of even
degree, see \cite{Koch-Sait}, and if we integrate only over a compact subset
$K$ of $\mathbb{Q}_{p}^{N}$ (the space cutoff), the function $\boldsymbol{\Phi
}_{H}$ is integrable, and we have a direct correspondence between
(\ref{Eq_A1}) and (\ref{Eq_A2}), i.e.
\[
\boldsymbol{\Phi}_{H}d\mu=\left\{  \frac{\exp\left(  -\int_{K}H^{\lozenge
}\left(  \boldsymbol{\Phi}\left(  x\right)  \right)  d^{N}x\right)  }{\int
\exp\left(  -\int_{K}H^{\lozenge}\left(  \boldsymbol{\Phi}\left(  x\right)
\right)  d^{N}x\right)  d\mu}\right\}  d\mu.
\]
In general the distribution $\boldsymbol{\Phi}_{H}$\ is not necessarily
positive, and for a large class of functions $H$, there are no measures
representing $\boldsymbol{\Phi}_{H}$. It turns out that $\boldsymbol{\Phi}%
_{H}$ can be represented by a measure if and only if $-H(it)+\frac{1}{2}t^{2}%
$, $t\in\mathbb{R}$, is a L\'{e}vy characteristic, see Theorem \ref{Theorem2A}%
. These measures are called generalized white noise measures.

Generalized white measures were considered in \cite{Zuniga-FAA-2017}, in the
$p$-adic framework, and in the Archimedean case in \cite{Albeverio-et-al-3}%
-\cite{Albeverio-et-al-4}. Euclidean random fields over $\mathbb{Q}_{p}^{N}$
were constructed by convolving generalized white noise with the fundamental
solutions of certain $p$-adic pseudodifferential equations. These fundamental
solutions are invariant under the action of a $p$-adic version of the
Euclidean group, see Section \ref{Sect_symmetries}.

For all convoluted generalized white noise measures such that their L\'{e}vy
characteristics have an analytic extension at the origin, we can give an
explicit formula for the generalized density with respect to the white noise
measure, see Theorem \ref{Theorem2}. In addition, there exists a large class
of distributions $\boldsymbol{\Phi}_{H}$ of type (\ref{Eq_A2}) that do not
have an associated measure, see Remark \ref{Nota_Theorem_3}. We also prove
that the Schwinger functions corresponding to convoluted generalized functions
\ satisfy Osterwalder-Schrader axioms (axioms OS1, OS2, OS4, OS5 in the
notation used in \cite{GS1999}) except for reflect positivity, see Lemma
\ref{Lemma1}, Theorems \ref{Theorem2}, \ref{Theorem3}, just like in the
Archimedean case presented in \cite{GS1999}.

The $p$-adic spacetime $\left(  \mathbb{Q}_{p}^{N},\mathfrak{q}\left(
\xi\right)  \right)  $ is a $\mathbb{Q}_{p}$-vector space of dimension $N$
with an elliptic quadratic form $\mathfrak{q}\left(  \xi\right)  $, i.e.
$\mathfrak{q}\left(  \xi\right)  =0\Leftrightarrow\xi=0$. This spacetime
differs from the classical spacetime $\left(  \mathbb{R}^{N},\xi_{1}%
^{2}+\cdots+\xi_{N}^{2}\right)  $ in several aspects. The $p$-adic spacetime
is not an `infinitely divisible continuum', because $\mathbb{Q}_{p}^{N}$ is a
completely disconnected topological space, the connected components (the
points) play the role of `spacetime quanta'. Since $\mathbb{Q}_{p}$ is not an
ordered field, the notions of past and future do not exist, then any $p$-adic
QFT is an acausal theory. The reader may consult the introduction of
\cite{Mendoza-Zuniga} for an in-depth discussion of this matter. Consequently,
the reflection positivity, if it exists in the $p$-adic framework, requires a
particular formulation, that we do not know at the moment. The study of the
$p$-adic Wightman functions via the reconstruction theorem is an open problem.

Another important difference between the classical case and the $p$-adic one
comes from the fact that in the $p$-adic setting there are no elliptic
quadratic forms in dimension $N\geq5$. We replace $\mathfrak{q}\left(
\xi\right)  $ by an elliptic polynomial $\mathfrak{l}\left(  \xi\right)  $,
which is a homogeneous polynomial satisfying $\mathfrak{l}\left(  \xi\right)
=0\Leftrightarrow\xi=0$. For any dimension $N$ there are elliptic polynomials
of degree $d\geq2$. We use $\left\vert \mathfrak{l}\left(  \xi\right)
\right\vert _{p}^{\frac{2}{d}}$ as a replacement of $\left\vert \mathfrak{q}%
\left(  \xi\right)  \right\vert _{p}$. This approach is particularly useful to
define the $p$-adic Laplace equation that the (free) covariance function
$C_{p}\left(  x-y\right)  $ satisfies, this equation has the following form:%
\[
\left(  \boldsymbol{L}_{\alpha}+m^{2}\right)  C_{p}\left(  x-y\right)
=\delta\left(  x-y\right)  \text{, \ }x\text{, }y\in\mathbb{Q}_{p}^{N},
\]
where $\alpha>0$, $m>0$ and $\boldsymbol{L}_{\alpha}$, is the
pseudodifferential operator
\[
\boldsymbol{L}_{\alpha}\varphi\left(  x\right)  =\mathcal{F}_{\xi\rightarrow
x}^{-1}(\left\vert \mathfrak{l}\left(  \xi\right)  \right\vert _{p}^{\alpha
}\mathcal{F}_{x\rightarrow\xi}\varphi),
\]
here $\mathcal{F}$ denotes the Fourier transform. The QFTs presented here are
families depending on several parameters, among them, $p$, $\alpha$, $m$,
$\mathfrak{l}\left(  \xi\right)  $.

The $p$-adic free covariance $C_{p}\left(  x-y\right)  $ may have
singularities at the origin depending on the parameters $\alpha$, $d$, $N$,
and has a `polynomial' decay at infinity, see Section \ref{Section_C_p}. The
$p$-adic cluster property holds under the condition $\alpha d>N$. Under this
hypothesis the covariance function \ does not have singularities at the
origin. Since $\alpha$ is a `free' parameter, this condition can be satisfied
in any dimension. We think that the condition $\alpha d>N$ is completely
necessary to have the cluster property due to the fact that our test functions
do not decay exponentially at infinity, see Remark \ref{Nota_Cluster}.

\section{\label{Sction_2}$p$\textbf{-}Adic Analysis: Essential Ideas}

In this section we collect some basic results about $p$-adic analysis that
will be used in the article. For an in-depth review of the $p$-adic analysis
the reader may consult \cite{A-K-S}, \cite{Taibleson}, \cite{V-V-Z}.

\subsection{The field of $p$-adic numbers}

Along this article $p$ will denote a prime number. The field of $p-$adic
numbers $%
\mathbb{Q}
_{p}$ is defined as the completion of the field of rational numbers
$\mathbb{Q}$ with respect to the $p-$adic norm $|\cdot|_{p}$, which is defined
as
\[
\left\vert x\right\vert _{p}=\left\{
\begin{array}
[c]{lll}%
0 & \text{if} & x=0\\
&  & \\
p^{-\gamma} & \text{if} & x=p^{\gamma}\frac{a}{b}\text{,}%
\end{array}
\right.
\]
where $a$ and $b$ are integers coprime with $p$. The integer $\gamma:=ord(x)
$, with $ord(0):=+\infty$, is called the\textit{\ }$p-$\textit{adic order of}
$x$.

Any $p-$adic number $x\neq0$ has a unique expansion of the form
\[
x=p^{ord(x)}\sum_{j=0}^{\infty}x_{j}p^{j},
\]
where $x_{j}\in\{0,\dots,p-1\}$ and $x_{0}\neq0$. By using this expansion, we
define \textit{the fractional part of }$x\in\mathbb{Q}_{p}$, denoted
$\{x\}_{p}$, as the rational number
\[
\left\{  x\right\}  _{p}=\left\{
\begin{array}
[c]{lll}%
0 & \text{if} & x=0\text{ or }ord(x)\geq0\\
&  & \\
p^{ord(x)}\sum_{j=0}^{-ord_{p}(x)-1}x_{j}p^{j} & \text{if} & ord(x)<0.
\end{array}
\right.
\]
In addition, any non-zero $p-$adic number can be represented uniquely as
$x=p^{ord(x)}ac\left(  x\right)  $ where $ac\left(  x\right)  =\sum
_{j=0}^{\infty}x_{j}p^{j}$, $x_{0}\neq0$, is called the \textit{angular
component} of $x$. Notice that $\left\vert ac\left(  x\right)  \right\vert
_{p}=1$.

We extend the $p-$adic norm to $%
\mathbb{Q}
_{p}^{N}$ by taking
\[
||x||_{p}:=\max_{1\leq i\leq N}|x_{i}|_{p},\text{ for }x=(x_{1},\dots
,x_{N})\in%
\mathbb{Q}
_{p}^{N}.
\]
We define $ord(x)=\min_{1\leq i\leq N}\{ord(x_{i})\}$, then $||x||_{p}%
=p^{-ord(x)}$.\ The metric space $\left(
\mathbb{Q}
_{p}^{N},||\cdot||_{p}\right)  $ is a complete ultrametric space. For
$r\in\mathbb{Z}$, denote by $B_{r}^{N}(a)=\{x\in%
\mathbb{Q}
_{p}^{N};||x-a||_{p}\leq p^{r}\}$ \textit{the ball of radius }$p^{r}$
\textit{with center at} $a=(a_{1},\dots,a_{N})\in%
\mathbb{Q}
_{p}^{N}$, and take $B_{r}^{N}(0):=B_{r}^{N}$. Note that $B_{r}^{N}%
(a)=B_{r}(a_{1})\times\cdots\times B_{r}(a_{N})$, where $B_{r}(a_{i}):=\{x\in%
\mathbb{Q}
_{p};|x_{i}-a_{i}|_{p}\leq p^{r}\}$ is the one-dimensional ball of radius
$p^{r}$ with center at $a_{i}\in%
\mathbb{Q}
_{p}$. The ball $B_{0}^{N}$ equals the product of $N$ copies of $B_{0}%
=\mathbb{Z}_{p}$, \textit{the ring of }$p-$\textit{adic integers of }$%
\mathbb{Q}
_{p}$. We also denote by $S_{r}^{N}(a)=\{x\in\mathbb{Q}_{p}^{N};||x-a||_{p}%
=p^{r}\}$ \textit{the sphere of radius }$p^{r}$ \textit{with center at}
$a=(a_{1},\dots,a_{N})\in%
\mathbb{Q}
_{p}^{N}$, and take $S_{r}^{N}(0):=S_{r}^{N}$. We notice that $S_{0}%
^{1}=\mathbb{Z}_{p}^{\times}$ (the group of units of $\mathbb{Z}_{p}$), but
$\left(  \mathbb{Z}_{p}^{\times}\right)  ^{N}\subsetneq S_{0}^{N}$. The balls
and spheres are both open and closed subsets in $%
\mathbb{Q}
_{p}^{N}$. In addition, two balls in $%
\mathbb{Q}
_{p}^{N}$ are either disjoint or one is contained in the other.

As a topological space $\left(
\mathbb{Q}
_{p}^{N},||\cdot||_{p}\right)  $ is totally disconnected, i.e. the only
connected \ subsets of $%
\mathbb{Q}
_{p}^{N}$ are the empty set and the points. A subset of $%
\mathbb{Q}
_{p}^{N}$ is compact if and only if it is closed and bounded in $%
\mathbb{Q}
_{p}^{N}$, see e.g. \cite[Section 1.3]{V-V-Z}, or \cite[Section 1.8]{A-K-S}.
The balls and spheres are compact subsets. Thus $\left(
\mathbb{Q}
_{p}^{N},||\cdot||_{p}\right)  $ is a locally compact topological space.

We will use $\Omega\left(  p^{-r}||x-a||_{p}\right)  $ to denote the
characteristic function of the ball $B_{r}^{N}(a)$. We will use the notation
$1_{A}$ for the characteristic function of a set $A$. Along the article
$d^{N}x$ will denote a Haar measure on $\left(
\mathbb{Q}
_{p}^{N},+\right)  $ normalized so that $\int_{%
\mathbb{Z}
_{p}^{N}}d^{N}x=1.$

\subsection{Some function spaces}

A complex-valued function $\varphi$ defined on $%
\mathbb{Q}
_{p}^{N}$ is \textit{called locally constant} if for any $x\in%
\mathbb{Q}
_{p}^{N}$ there exist an integer $l(x)\in\mathbb{Z}$ such that
\[
\varphi(x+x^{\prime})=\varphi(x)\text{ for }x^{\prime}\in B_{l(x)}^{N}.
\]
A function $\varphi:%
\mathbb{Q}
_{p}^{N}\rightarrow\mathbb{C}$ is called a \textit{Bruhat-Schwartz function
(or a test function)} if it is locally constant with compact support. The
$\mathbb{C}$-vector space of Bruhat-Schwartz functions is denoted by
$\mathcal{D}:=\mathcal{D}(%
\mathbb{Q}
_{p}^{N})$. Let $\mathcal{D}^{\prime}:=\mathcal{D}^{\prime}(%
\mathbb{Q}
_{p}^{N})$ denote the set of all continuous functional (distributions) on
$\mathcal{D}$.

We will denote by $\mathcal{D}_{\mathbb{R}}:=\mathcal{D}_{\mathbb{R}}(%
\mathbb{Q}
_{p}^{N})$, the $\mathbb{R}$-vector space of test functions, and by
$\mathcal{D}_{\mathbb{R}}^{\prime}:=\mathcal{D}_{\mathbb{R}}^{\prime}(%
\mathbb{Q}
_{p}^{N})$, the $\mathbb{R}$-vector space of distributions.

Given $\rho\in\lbrack0,\infty)$, we denote by $L^{\rho}:=L^{\rho}\left(
\mathbb{Q}
_{p}^{N}\right)  :=L^{\rho}\left(
\mathbb{Q}
_{p}^{N},d^{N}x\right)  ,$ the $%
\mathbb{C}
-$vector space of all the complex valued functions $g$ satisfying $\int_{%
\mathbb{Q}
_{p}^{N}}\left\vert g\left(  x\right)  \right\vert ^{\rho}d^{N}x<\infty$, and
$L^{\infty}\allowbreak:=L^{\infty}\left(
\mathbb{Q}
_{p}^{N}\right)  =L^{\infty}\left(
\mathbb{Q}
_{p}^{N},d^{N}x\right)  $ denotes the $%
\mathbb{C}
-$vector space of all the complex valued functions $g$ such that the essential
supremum of $|g|$ is bounded. The corresponding $\mathbb{R}$-vector spaces are
denoted as $L_{\mathbb{R}}^{\rho}\allowbreak:=L_{\mathbb{R}}^{\rho}\left(
\mathbb{Q}
_{p}^{N}\right)  =L_{\mathbb{R}}^{\rho}\left(
\mathbb{Q}
_{p}^{N},d^{N}x\right)  $, $1\leq\rho\leq\infty$.

Set
\[
\mathcal{C}_{0}(%
\mathbb{Q}
_{p}^{N},\mathbb{C}):=\left\{  f:%
\mathbb{Q}
_{p}^{N}\rightarrow%
\mathbb{C}
;\text{ }f\text{ is continuous and }\lim_{||x||_{p}\rightarrow\infty
}f(x)=0\right\}  ,
\]
where $\lim_{||x||_{p}\rightarrow\infty}f(x)=0$ means that for every
$\epsilon>0$ there exists a compact subset $B(\epsilon)$ such that
$|f(x)|<\epsilon$ for $x\in%
\mathbb{Q}
_{p}^{N}\backslash B(\epsilon).$ We recall that $(\mathcal{C}_{0}(%
\mathbb{Q}
_{p}^{N},\mathbb{C}),||\cdot||_{L^{\infty}})$ is a Banach space. The
corresponding $\mathbb{R}$-vector space will be denoted as $\mathcal{C}_{0}(%
\mathbb{Q}
_{p}^{N},\mathbb{R})$.

\subsection{Fourier transform}

Set $\chi_{p}(y):=\exp(2\pi i\{y\}_{p})$ for $y\in%
\mathbb{Q}
_{p}$. The map $\chi_{p}(\cdot)$ is an additive character on $%
\mathbb{Q}
_{p}$, i.e. a continuous map from $\left(
\mathbb{Q}
_{p},+\right)  $ into $S$ (the unit circle considered as multiplicative group)
satisfying $\chi_{p}(x_{0}+x_{1})=\chi_{p}(x_{0})\chi_{p}(x_{1})$,
$x_{0},x_{1}\in%
\mathbb{Q}
_{p}$. The additive characters of $%
\mathbb{Q}
_{p}$ form an Abelian group which is isomorphic to $\left(
\mathbb{Q}
_{p},+\right)  $, the isomorphism is given by $\xi\rightarrow\chi_{p}(\xi x)$,
see e.g. \cite[Section 2.3]{A-K-S}.

Given $x=(x_{1},\dots,x_{N}),$ $\xi=(\xi_{1},\dots,\xi_{N})\in%
\mathbb{Q}
_{p}^{N}$, we set $x\cdot\xi:=\sum_{j=1}^{N}x_{j}\xi_{j}$. If $f\in L^{1}$ its
Fourier transform is defined by
\[
(\mathcal{F}f)(\xi)=\int_{%
\mathbb{Q}
_{p}^{N}}\chi_{p}(\xi\cdot x)f(x)d^{N}x,\quad\text{for }\xi\in%
\mathbb{Q}
_{p}^{N}.
\]
We will also use the notation $\mathcal{F}_{x\rightarrow\xi}f$ and
$\widehat{f}$\ for the Fourier transform of $f$. The Fourier transform is a
linear isomorphism from $\mathcal{D}(%
\mathbb{Q}
_{p}^{N})$ onto itself satisfying
\begin{equation}
(\mathcal{F}(\mathcal{F}f))(\xi)=f(-\xi), \label{FF(f)}%
\end{equation}
for every $f\in\mathcal{D}(%
\mathbb{Q}
_{p}^{N}),$ see e.g. \cite[Section 4.8]{A-K-S}. If $f\in L^{2},$ its Fourier
transform is defined as
\[
(\mathcal{F}f)(\xi)=\lim_{k\rightarrow\infty}\int_{||x||_{p}\leq p^{k}}%
\chi_{p}(\xi\cdot x)f(x)d^{N}x,\quad\text{for }\xi\in%
\mathbb{Q}
_{p}^{N},
\]
where the limit is taken in $L^{2}.$ We recall that the Fourier transform is
unitary on $L^{2},$ i.e. $||f||_{L^{2}}=||\mathcal{F}f||_{L^{2}}$ for $f\in
L^{2}$ and that (\ref{FF(f)}) is also valid in $L^{2}$, see e.g. \cite[Chapter
III, Section 2]{Taibleson}.

The Fourier transform $\mathcal{F}\left[  W\right]  $ of a distribution
$W\in\mathcal{D}^{\prime}\left(
\mathbb{Q}
_{p}^{N}\right)  $ is defined by%
\[
\left(  \mathcal{F}\left[  W\right]  ,\varphi\right)  =\left(  W,\mathcal{F}%
\left[  \varphi\right]  \right)  \text{ for all }\varphi\in\mathcal{D}(%
\mathbb{Q}
_{p}^{N})\text{.}%
\]
The Fourier transform $W\rightarrow\mathcal{F}\left[  W\right]  $ is a linear
isomorphism from $\mathcal{D}^{\prime}\left(
\mathbb{Q}
_{p}^{N}\right)  $\ onto itself. Furthermore, $W=\mathcal{F}\left[
\mathcal{F}\left[  W\right]  \left(  -\xi\right)  \right]  $. We also use the
notation $\mathcal{F}_{x\rightarrow\xi}W$ and $\widehat{W}$ for the Fourier
transform of $W.$

\section{$p$-Adic White Noise}

In this section we review some basic aspects of the white noise calculus in
the $p$-adic setting. For a in-depth exposition on the white noise calculus on
arbitrary nuclear spaces the reader may consult \cite{Ber-Kon},
\cite{Gelfan-Vilenkin}, \cite{Hida et al}, \cite{Huang-Yang}, \cite{Obata}. We
will use white noise calculus on the nuclear spaces $\mathcal{H}_{\infty}$
introduced by Z\'{u}\~{n}iga-Galindo in \cite{Zuniga-FAA-2017}, see also
\cite[Chapters 10, 11]{KKZuniga}.

\subsection{A class of non-Archimedean nuclear spaces}

\subsubsection{\label{Section1}$\mathcal{H}_{\infty}$, a non-Archimedean
analog of the Schwartz space}

We denote the set on non-negative integers by $\mathbb{N}$, and set $\left[
\xi\right]  _{p}:=[\max(1,\Vert\xi\Vert_{p})]$ for $\xi\in\mathbb{Q}_{p}^{N}$.
We define for $\varphi$, $\theta\in\mathcal{D}(\mathbb{Q}_{p}^{N})$, and
$l\in\mathbb{N}$, the following scalar product:%
\[
\left\langle \varphi,\theta\right\rangle _{l}=\int_{\mathbb{Q}_{p}^{N}}\left[
\xi\right]  _{p}^{l}\overline{\widehat{\varphi}\left(  \xi\right)  }%
\widehat{\theta}\left(  \xi\right)  d^{N}\xi\text{,}%
\]
where the overbar denotes the complex conjugate. We also set $\left\Vert
\varphi\right\Vert _{l}:=\left\langle \varphi,\varphi\right\rangle _{l}$.
Notice that $\left\Vert \cdot\right\Vert _{l}\leq\left\Vert \cdot\right\Vert
_{m}$ for $l\leq m$. We denote by $\mathcal{H}_{l}(\mathbb{C}):=$
$\mathcal{H}_{l}(\mathbb{Q}_{p}^{N},\mathbb{C})$ the complex Hilbert space
obtained by completing $\mathcal{D}(\mathbb{Q}_{p}^{N})$ with respect to
$\left\langle \cdot,\cdot\right\rangle _{l}$. Then $\mathcal{H}_{m}%
(\mathbb{C})\hookrightarrow\mathcal{H}_{l}(\mathbb{C})$ for $l\leq m$. Now we
set
\[
\mathcal{H}_{\infty}(\mathbb{C}):=\mathcal{H}_{\infty}(\mathbb{Q}_{p}%
^{N},\mathbb{C})=%
{\displaystyle\bigcap\nolimits_{l\in\mathbb{N}}}
\mathcal{H}_{l}(\mathbb{C}).
\]
Notice that $\mathcal{H}_{\infty}(\mathbb{C})\subset L^{2}$. With the topology
induced by the family of seminorms $\left\{  \left\Vert \cdot\right\Vert
_{l}\right\}  _{l\in\mathbb{N}}$, $\mathcal{H}_{\infty}(\mathbb{C})$ becomes a
locally convex space, which is metrizable. Indeed,%
\[
d(f,g):=\max_{l\in\mathbb{N}}\left\{  2^{-l}\frac{\left\Vert f-g\right\Vert
_{l}}{1+\left\Vert f-g\right\Vert _{l}}\right\}  \text{, with }f\text{, }%
g\in\mathcal{H}_{\infty}(\mathbb{C})\text{, }%
\]
is a metric for the topology of $\mathcal{H}_{\infty}(\mathbb{C})$. The
projective topology $\tau_{P}$ of $\mathcal{H}_{\infty}(\mathbb{C})$ coincides
with the topology induced by the family of seminorms $\left\{  \left\Vert
\cdot\right\Vert _{l}\right\}  _{l\in\mathbb{N}}$. The space $\mathcal{H}%
_{\infty}(\mathbb{C})$ endowed with the topology $\tau_{P}$ is a countably
Hilbert space in the sense of Gel'fand-Vilenkin. Furthermore, $\left(
\mathcal{H}_{\infty}(\mathbb{C}),\tau_{P}\right)  $ is metrizable and complete
and hence a Fr\'{e}chet space, cf. \cite[Lemma 10.3]{KKZuniga}, see also
\cite{Zuniga-FAA-2017}.

The space $(\mathcal{H}_{\infty}(\mathbb{C}),d)$ is the completion of
$(\mathcal{D}(\mathbb{Q}_{p}^{N}),d)$ with respect to $d$, and since
$\mathcal{D}(\mathbb{Q}_{p}^{N})$ is nuclear, then $\mathcal{H}_{\infty
}(\mathbb{C})$ is a nuclear space, which is continuously embedded in
$C_{0}(\mathbb{Q}_{p}^{N},\mathbb{C})$, the space of complex-valued bounded
functions vanishing at infinity. In addition, $\mathcal{H}_{\infty}%
(\mathbb{C})\subset L^{1}\cap L^{2}$, cf. \cite[Theorem 10.15]{KKZuniga}.

\begin{remark}
(i) We denote by $\mathcal{H}_{l}(\mathbb{R}):=\mathcal{H}_{l}(\mathbb{Q}%
_{p}^{N},\mathbb{R})$ the real Hilbert space obtained by completing
$\mathcal{D}_{\mathbb{R}}(\mathbb{Q}_{p}^{N})$ with respect to $\left\langle
\cdot,\cdot\right\rangle _{l}$. We also set $\mathcal{H}_{\infty}%
(\mathbb{Q}_{p}^{N},\mathbb{R}):=\mathcal{H}_{\infty}(\mathbb{R})=\cap
_{l\in\mathbb{N}}\mathcal{H}_{l}(\mathbb{R})$. In the case in which the ground
field ($\mathbb{R}$ or $\mathbb{C)}$ is clear, we shall use the simplified
notation $\mathcal{H}_{l}$, $\mathcal{H}_{\infty}$. All the above announced
results \ for the spaces $\mathcal{H}_{l}(\mathbb{C})$, $\mathcal{H}_{\infty
}(\mathbb{C})$ are valid for the spaces $\mathcal{H}_{l}(\mathbb{R})$,
$\mathcal{H}_{\infty}(\mathbb{R})$. In particular, $\mathcal{H}_{\infty
}(\mathbb{R})$ is a nuclear countably Hilbert space.

(ii) The following characterization of the space $\mathcal{H}_{\infty
}(\mathbb{C})$ is very useful:%
\begin{align*}
\mathcal{H}_{\infty}(\mathbb{C})  &  =\left\{  f\in L^{2}\left(
\mathbb{Q}_{p}^{N}\right)  ;\left\Vert f\right\Vert _{l}<\infty\text{ for any
}l\in\mathbb{N}\right\} \\
&  =\left\{  W\in\mathcal{D}^{\prime}\left(  \mathbb{Q}_{p}^{N}\right)
;\left\Vert W\right\Vert _{l}<\infty\text{ for any }l\in\mathbb{N}\right\}  ,
\end{align*}
cf. \cite[Lemma 10.8]{KKZuniga}. An analog result is valid for $\mathcal{H}%
_{\infty}(\mathbb{R})$.

(iii) The spaces $\mathcal{H}_{l}(\mathbb{R})$, $\mathcal{H}_{l}(\mathbb{C})$,
for any $l\in\mathbb{N}$, are nuclear and consequently they are separable, cf.
\cite[Chapter I, Section 3.4]{Gelfan-Vilenkin}.
\end{remark}

The spaces $\mathcal{H}_{\infty}(\mathbb{Q}_{p}^{N},\mathbb{C})\mathcal{\ }%
$and $\mathcal{H}_{\infty}(\mathbb{Q}_{p}^{N},\mathbb{R})$ were introduced in
\cite{Zuniga-FAA-2017}, see also \cite{KKZuniga}. These spaces are invariant
under the action of a large class of pseudodifferential operators.

\subsubsection{The dual space of $\mathcal{H}_{\infty}$}

For $m\in\mathbb{N}$, and $W\in\mathcal{D}^{\prime}\left(  \mathbb{Q}_{p}%
^{N}\right)  $ such that $\widehat{W}$ is a measurable function, we set%
\[
\left\Vert W\right\Vert _{-m}^{2}:=\int_{\mathbb{Q}_{p}^{N}}\left[
\xi\right]  _{p}^{-m}\left\vert \widehat{W}\left(  \xi\right)  \right\vert
^{2}d^{N}\xi\text{.}%
\]
Then%
\begin{equation}
\mathcal{H}_{-m}(\mathbb{C}):=\mathcal{H}_{-m}(\mathbb{Q}_{p}^{N}%
,\mathbb{C})=\left\{  W\in\mathcal{D}^{\prime}\left(  \mathbb{Q}_{p}%
^{N}\right)  ;\left\Vert W\right\Vert _{-m}<\infty\right\}  \label{Eq_A}%
\end{equation}
is a complex Hilbert space. If $\mathcal{X}$ is a locally convex, we denote by
$\mathcal{X}^{\ast}$ the dual space endowed with the strong dual topology or
the topology of the bounded convergence. We denote by $\mathcal{H}_{m}^{\ast
}(\mathbb{C})$ the dual of $\mathcal{H}_{m}(\mathbb{C})$ for $m\in\mathbb{N}$,
we identify $\mathcal{H}_{m}^{\ast}(\mathbb{C})$ with $\mathcal{H}%
_{-m}(\mathbb{C})$, by using the bilinear form:%
\begin{equation}
\left\langle W,g\right\rangle =\int_{\mathbb{Q}_{p}^{N}}\overline{\widehat
{W}\left(  \xi\right)  }\widehat{g}\left(  \xi\right)  d^{N}\xi\text{ for
}W\in\mathcal{H}_{-m}(\mathbb{C})\text{ and }g\in\mathcal{H}_{m}%
(\mathbb{C})\text{.} \label{pairing}%
\end{equation}
Then%
\begin{align*}
\mathcal{H}_{\infty}^{\ast}(\mathbb{Q}_{p}^{N},\mathbb{C})  &  :=\mathcal{H}%
_{\infty}^{\ast}(\mathbb{C})=\bigcup\limits_{m\in\mathbb{N}}\mathcal{H}%
_{-m}(\mathbb{C})\\
&  =\left\{  W\in\mathcal{D}^{\prime}\left(  \mathbb{Q}_{p}^{N}\right)
;\left\Vert W\right\Vert _{-m}<\infty\text{ for some }m\in\mathbb{N}\right\}
.
\end{align*}
We consider $\mathcal{H}_{\infty}^{\ast}(\mathbb{C})$ endowed with the strong
topology. We use (\ref{pairing}) as pairing between $\mathcal{H}_{\infty
}^{\ast}(\mathbb{C})$ and $\mathcal{H}_{\infty}(\mathbb{C})$. By a similar
construction one obtains the space $\mathcal{H}_{\infty}^{\ast}(\mathbb{R}%
):=\mathcal{H}_{\infty}^{\ast}(\mathbb{Q}_{p}^{N},\mathbb{R})$. The above
announced results are also valid for $\mathcal{H}_{\infty}^{\ast}(\mathbb{R}%
)$. If there is no danger of confusion we use $\mathcal{H}_{\infty}^{\ast}$
instead of $\mathcal{H}_{\infty}^{\ast}(\mathbb{C})$ or $\mathcal{H}_{\infty
}^{\ast}(\mathbb{R})$.

\begin{remark}
\label{Nota1}(i) For complex and real spaces, $\left\Vert \cdot\right\Vert
_{\pm l}$ denotes the norm on $\mathcal{H}_{l}$ and $\mathcal{H}_{-l}$. We
denote by $\left\langle \cdot,\cdot\right\rangle $ the dual pairings between
$\mathcal{H}_{-l}$ and $\mathcal{H}_{l}$ and between $\mathcal{H}_{\infty}$
and $\mathcal{H}_{\infty}^{\ast}$. We preserve this notation for the norm and
pairing on tensor powers of these spaces.

(ii) If $\left\{  \mathcal{X}_{l}\right\}  _{l\in A}$ is a family of locally
convex spaces, we denote by $\underleftarrow{\lim}_{l\in\mathbb{N}}%
\mathcal{X}_{l}$ the projective limit of the family, and by $\underrightarrow
{\lim}_{l\in\mathbb{N}}\mathcal{X}_{l}$ the inductive limit of the family.

(iii) If $\mathcal{N}$ is a nuclear space, which is the projective limit of
the Hilbert spaces $H_{l}$, $l\in\mathbb{N}$,the $n$-th symmetric tensor
product of $\mathcal{N}$ is defined as $\mathcal{N}^{\widehat{\otimes}%
n}=\underleftarrow{\lim}_{l\in\mathbb{N}}H_{l}^{\widehat{\otimes}n}$. This is
a nuclear space. The dual space is $\mathcal{N}^{\ast\widehat{\otimes}%
n}=\underrightarrow{\lim}_{l\in\mathbb{N}}H_{-l}^{\widehat{\otimes}n}$.
\end{remark}

\subsection{Non-Archimedean Gaussian measures}

The spaces
\[
\mathcal{H}_{\infty}(\mathbb{R})\hookrightarrow L_{\mathbb{R}}^{2}\left(
\mathbb{Q}_{p}^{N}\right)  \hookrightarrow\mathcal{H}_{\infty}^{\ast
}(\mathbb{R})
\]
form a Gel'fand triple, that is, $\mathcal{H}_{\infty}(\mathbb{R})$ is a
nuclear countably Hilbert space which is densely and continuously embedded in
$L_{\mathbb{R}}^{2}$ and $\left\Vert g\right\Vert _{0}^{2}=\left\langle
g,g\right\rangle _{0}$ for $g\in\mathcal{H}_{\infty}(\mathbb{R})$. This triple
was introduced in \cite{Zuniga-FAA-2017}, see also \cite[Chapter 10]%
{KKZuniga}. The inner product and the norm of $\left(  L_{\mathbb{R}}%
^{2}\left(  \mathbb{Q}_{p}^{N}\right)  \right)  ^{\otimes m}\simeq
L_{\mathbb{R}}^{2}\left(  \mathbb{Q}_{p}^{Nm}\right)  $ are denoted by
$\left\langle \cdot,\cdot\right\rangle _{0}$ and $\left\Vert \cdot\right\Vert
_{0}$. From now on, we consider $\mathcal{H}_{\infty}^{^{\widehat{\otimes}n}%
}\left(  \mathbb{R}\right)  $ as subspace of $\mathcal{H}_{\infty}^{^{\otimes
n}}\left(  \mathbb{R}\right)  $, then $\left\langle \cdot,\cdot\right\rangle
_{\mathcal{H}_{\infty}^{^{\widehat{\otimes}n}}\left(  \mathbb{R}\right)
}=n!\left\langle \cdot,\cdot\right\rangle _{0}$.

We denote by $\mathcal{B}:=\mathcal{B}(\mathcal{H}_{\infty}^{\ast}%
(\mathbb{R}))$ the $\sigma$-algebra generated by the cylinder subsets of
$\mathcal{H}_{\infty}^{\ast}(\mathbb{R})$. The mapping
\[%
\begin{array}
[c]{cccc}%
\mathcal{C}: & \mathcal{H}_{\infty}(\mathbb{R}) & \rightarrow & \mathbb{C}\\
& f & \rightarrow & e^{-\frac{1}{2}\left\Vert f\right\Vert _{0}^{2}}%
\end{array}
\]
defines a characteristic functional, i.e. $\mathcal{C}$ is continuous,
positive definite and $\mathcal{C}\left(  0\right)  =1$. By the Bochner-Minlos
theorem, see e.g. \cite{Ber-Kon}, \cite{Hida et al}, there exists a
probability measure $\mu$, called \textit{the canonical Gaussian measure} on
$\left(  \mathcal{H}_{\infty}^{\ast}(\mathbb{R}),\mathcal{B}\right)  $, given
by its characteristic functional as%

\[
\int_{\mathcal{H}_{\infty}^{\ast}(\mathbb{R})}e^{i\langle W,f\rangle}%
d\mu(W)=e^{-\frac{1}{2}\left\Vert f\right\Vert _{^{0}}^{2}}\text{,}%
\ \ f\in\mathcal{H}_{\infty}(\mathbb{R})\text{.}%
\]

We set $\left(  L_{\mathbb{C}}^{2}\right)  :=L^{2}\left(  \mathcal{H}_{\infty
}^{\ast}(\mathbb{R}),\mu;\mathbb{C}\right)  $ to denote the complex vector
space of measu\-rable functions $\Psi:\mathcal{H}_{\infty}^{\ast}%
(\mathbb{R})\rightarrow\mathbb{C}$ satisfying%
\[
\left\Vert \Psi\right\Vert _{\left(  L_{\mathbb{C}}^{2}\right)  }^{2}%
=\int_{\mathcal{H}_{\infty}^{\ast}(\mathbb{R})}\left\vert \Psi\left(
W\right)  \right\vert ^{2}d\mu(W)<\infty\text{.}%
\]
The space $\left(  L_{\mathbb{R}}^{2}\right)  :=L^{2}\left(  \mathcal{H}%
_{\infty}^{\ast}(\mathbb{R}),\mu;\mathbb{R}\right)  $ is defined in a similar
way. The pairing $\mathcal{H}_{\infty}^{\ast}(\mathbb{R})\times\mathcal{H}%
_{\infty}(\mathbb{R})$ can be extended to $\mathcal{H}_{\infty}^{\ast
}(\mathbb{R})\times L^{2}(\mathbb{Q}_{p}^{N})$ as an $\left(  L_{\mathbb{C}%
}^{2}\right)  $-function on $\mathcal{H}_{\infty}^{\ast}(\mathbb{R})$, this
fact follows from
\begin{equation}
\int_{\mathcal{H}_{\infty}^{\ast}(\mathbb{R})}\left\vert \left\langle
W,g\right\rangle \right\vert ^{2}d\mu(W)=\left\Vert g\right\Vert _{0}^{2},
\label{Eq_3}%
\end{equation}
see e.g. \cite[Lemma 2.1.5]{Obata}. If $g\in L_{\mathbb{R}}^{2}$, then
$W\rightarrow\left\langle W,g\right\rangle $ belongs to $\left(
L_{\mathbb{R}}^{2}\right)  $.

Let $f\in\mathcal{H}_{\infty}(\mathbb{R})$ and $W_{f}(J):=\left\langle
J,f\right\rangle $, $J\in$ $\mathcal{H}_{\infty}^{\ast}(\mathbb{R})$. Then
$W_{f}$ is a Gaussian random variable on $\left(  \mathcal{H}_{\infty}^{\ast
}(\mathbb{R}),\mu\right)  $ satisfying
\[
\mathbb{E}_{\mu}(W_{f})=0\text{, \ }\mathbb{E}_{\mu}(W_{f}^{2})=\left\Vert
f\right\Vert _{0}^{2}.
\]
Then the linear map
\[%
\begin{array}
[c]{ccc}%
\mathcal{H}_{\infty}(\mathbb{R}) & \rightarrow & \left(  L_{\mathbb{R}}%
^{2}\right) \\
&  & \\
f & \rightarrow & W_{f}%
\end{array}
\]
can be extended to a linear isometry from $L^{2}(\mathbb{Q}_{p}^{N})$\ to
$\left(  L_{\mathbb{C}}^{2}\right)  $.

\subsection{Wick-ordered polynomials}

Let $\mathcal{P}_{n}(\mathbb{R})$, respectively $\mathcal{P}_{n}(\mathbb{C})$,
be the vector space of finite linear combinations of functions of the form%
\[
W\rightarrow\left\langle W,f\right\rangle ^{n}=\left\langle W^{\otimes
n},f^{\otimes n}\right\rangle \text{, with }W\in\mathcal{H}_{\infty}^{\ast
}(\mathbb{R})\text{,}%
\]
where $f$ runs over $\mathcal{H}_{\infty}(\mathbb{R})$, respectively
$\mathcal{H}_{\infty}(\mathbb{C})$. Notice that $\mathcal{P}_{n}%
(\mathbb{C})=\mathcal{P}_{n}(\mathbb{R})+i\mathcal{P}_{n}(\mathbb{R})$. An
element of the direct algebraic sums%
\[
\mathcal{P}(\mathbb{R}):=\bigoplus\nolimits_{n=0}^{\infty}\mathcal{P}%
_{n}(\mathbb{R})\text{, \ }\mathcal{P}(\mathbb{C}):=\bigoplus\nolimits_{n=0}%
^{\infty}\mathcal{P}_{n}(\mathbb{C})\text{\ }%
\]
is called a \textit{polynomial }on the Gaussian space $\mathcal{H}_{\infty
}^{\ast}(\mathbb{R})$. These functions are not very useful because they do not
satisfy orthogonality relations. This is the main motivation to introduce and
utilize the Wick-ordered polynomials.

For $W\in\mathcal{H}_{\infty}^{\ast}\left(  \mathbb{R}\right)  $ and
$f\in\mathcal{H}_{\infty}$, we define \textit{the Wick-ordered monomial} as%
\begin{align*}
\left\langle :W^{\otimes n}:,f^{\otimes n}\right\rangle  &  =\sum
\limits_{k=0}^{\left[  \frac{n}{2}\right]  }\frac{n!}{k!\left(  n-2k\right)
!}\left(  \frac{-1}{2}\left\langle f,f\right\rangle _{0}\right)
^{k}\left\langle W,f\right\rangle ^{n-2k}\\
&  =\left\Vert f\right\Vert _{0}^{n}\boldsymbol{H}_{n}\left(  \left\Vert
f\right\Vert _{0}^{-1}\left\langle W,f\right\rangle \right)  ,
\end{align*}
where $\boldsymbol{H}_{n}$ denotes the $n$-th Hermite polynomial. Then
$:W^{\otimes n}:\in\mathcal{H}_{\infty}^{\ast\widehat{\otimes}n}$, in
addition, any polynomial $\Phi\in\mathcal{P}(\mathbb{R})$, respectively
$\mathcal{P}(\mathbb{C})$, is expressed as%
\begin{equation}
\Phi\left(  W\right)  =\sum\limits_{n=0}^{\infty}\left\langle :W^{\otimes
n}:,\phi_{n}\right\rangle , \label{Eq_4}%
\end{equation}
where $\phi_{n}$ belong to the symmetric $n$-fold algebraic tensor product
$\left(  \mathcal{H}_{\infty}(\mathbb{R})\right)  ^{\widehat{\otimes}n}$\ of
$\mathcal{H}_{\infty}(\mathbb{R})$, respectively of $\mathcal{H}_{\infty
}(\mathbb{C})$, and the sum symbol involves only a finite number of non-zero
terms. \ A function of type (\ref{Eq_4}) \ is called a \textit{Wick-ordered
polynomial}. For two polynomials $\Phi$, $\Psi\in\mathcal{P}(\mathbb{C})$
given respectively by (\ref{Eq_4}) with $\phi_{n}\in\left(  \mathcal{H}%
_{\infty}(\mathbb{C})\right)  ^{\widehat{\otimes}n}$, and by%
\begin{equation}
\Psi\left(  W\right)  =\sum\limits_{n=0}^{\infty}\left\langle :W^{\otimes
n}:,\psi_{n}\right\rangle ,\text{ with }\psi_{n}\in\left(  \mathcal{H}%
_{\infty}(\mathbb{C})\right)  ^{\widehat{\otimes}n}\text{,} \label{Eq_5}%
\end{equation}
it holds that
\[
\int_{\mathcal{H}_{\infty}^{\ast}(\mathbb{R})}\Phi\left(  W\right)
\Psi\left(  W\right)  d\mu(W)=\sum\nolimits_{n=0}^{\infty}n!\left\langle
\varphi_{n},\psi_{n}\right\rangle _{0},
\]
where $\left\langle \cdot,\cdot\right\rangle _{0}$ denotes the scalar product
in $\left(  L^{2}\left(  \mathbb{Q}_{p}^{N}\right)  \right)  ^{\otimes n}$. In
parti\-cular,
\[
\left\Vert \Phi\right\Vert _{\left(  L_{\mathbb{C}}^{2}\right)  }^{2}%
=\sum\nolimits_{n=0}^{\infty}n!\left\Vert \phi_{n}\right\Vert _{0}^{2},
\]
where $\left\Vert \cdot\right\Vert _{0}$ denotes the norms of $\left(
L^{2}\left(  \mathbb{Q}_{p}^{N}\right)  \right)  ^{\otimes n}$, see e.g.
\cite[Proposition 2.2.10]{Obata}. Consequently, each $\Psi\in\mathcal{P}%
(\mathbb{C})$ is uniquely expressed as a Wick-ordered polynomial.

\begin{remark}
\label{Nota2}We denote by $I_{n}(f_{n})$ the linear extension to $\left(
L^{2}\left(  \mathbb{Q}_{p}^{N}\right)  \right)  ^{\widehat{\otimes}n}$ of the
map $f_{n}\rightarrow\left\langle :W^{\otimes n}:,f_{n}\right\rangle $,
$W\in\mathcal{H}_{\infty}^{\ast}\left(  \mathbb{R}\right)  $, then%
\[
I_{n}(f^{\otimes n})=\left\Vert f\right\Vert _{0}^{n}\boldsymbol{H}%
_{n}(\left\Vert f\right\Vert _{0}^{-1}W_{f})\text{, \ }f\in L^{2},
\]
and
\[
\int_{\mathcal{H}_{\infty}^{\ast}(\mathbb{R})}I_{n}(f_{n})I_{m}(g_{m}%
)d\mu=\delta_{nm}n!\left\langle f_{n},g_{m}\right\rangle _{0}\text{, \ }%
f_{n}\in L^{2\widehat{\otimes}n}\text{, }g_{m}\in L^{2\widehat{\otimes}m}.
\]
We shall also use $\left\langle :W^{\otimes n}:,f_{n}\right\rangle $ to denote
$I_{n}(f_{n})$ formally. In this case the symbol $\left\langle \cdot
,\cdot\right\rangle $ should not be confused with the bilinear form on
$\mathcal{H}_{\infty}^{\ast}\times\mathcal{H}_{\infty}$.
\end{remark}

\subsection{Wiener-It\^{o}-Segal isomorphism}

Let $\Gamma\left(  L^{2}\left(  \mathbb{Q}_{p}^{N}\right)  \right)  $ be the
space of sequences $\boldsymbol{f}=\left\{  f_{n}\right\}  _{n\in\mathbb{N}}$,
$f_{n}\in\left(  L^{2}\left(  \mathbb{Q}_{p}^{N}\right)  \right)
^{\widehat{\otimes}n}$, such that
\[
\left\Vert \boldsymbol{f}\right\Vert _{\Gamma\left(  L^{2}\left(
\mathbb{Q}_{p}^{N}\right)  \right)  }^{2}:=\sum\nolimits_{n=0}^{\infty
}n!\left\Vert f_{n}\right\Vert _{0}^{2}<\infty.
\]
The Hilbert space $\Gamma\left(  L^{2}\left(  \mathbb{Q}_{p}^{N}\right)
\right)  $ is called \textit{the Boson Fock Space on} $L^{2}\left(
\mathbb{Q}_{p}^{N}\right)  $. The Wiener-It\^{o}-Segal theorem asserts that
for each $\Phi\in\left(  L_{\mathbb{C}}^{2}\right)  $ there exists a sequence
$\boldsymbol{\phi}=\left\{  \phi_{n}\right\}  _{n\in\mathbb{N}}$ in
$\Gamma\left(  L^{2}\left(  \mathbb{Q}_{p}^{N}\right)  \right)  $ such that
(\ref{Eq_4}) holds in the $\left(  L_{\mathbb{C}}^{2}\right)  $-sense, but
with $\phi_{n}\in\left(  L^{2}\left(  \mathbb{Q}_{p}^{N}\right)  \right)
^{\widehat{\otimes}n}$, see Remark \ref{Nota2}. Conversely, for any
$\boldsymbol{\phi}=\left\{  \phi_{n}\right\}  _{n\in\mathbb{N}}\in
\Gamma\left(  L_{\mathbb{C}}^{2}\left(  \mathbb{Q}_{p}^{N}\right)  \right)  $,
(\ref{Eq_4}) defines a function in $\left(  L_{\mathbb{C}}^{2}\right)  $. In
this case
\[
\left\Vert \Phi\right\Vert _{\left(  L_{\mathbb{C}}^{2}\right)  }^{2}%
=\sum\limits_{n=0}^{\infty}n!\left\Vert \phi_{n}\right\Vert _{0}%
^{2}=\left\Vert \boldsymbol{\phi}\right\Vert _{\Gamma\left(  L^{2}\left(
\mathbb{Q}_{p}^{N}\right)  \right)  }^{2},
\]
see e.g. \cite[Theorem 2.3.5]{Obata}, \cite{Segal}

\section{\label{Section_4}Non-Archimedean Kondratiev Spaces of Test Functions
and Distributions}

In this section we introduce non-Archimedean versions of Kondratiev-type
spaces of test functions and distributions.

\subsection{Kondratiev-type spaces of test functions}

We define for $l$, $k\in\mathbb{N}$, and $\beta\in\left[  0,1\right]  $ fixed,
the following norm on $\left(  L_{\mathbb{C}}^{2}\right)  $:%
\[
\left\Vert \Phi\right\Vert _{l,k,\beta}^{2}=\sum\limits_{n=0}^{\infty}\left(
n!\right)  ^{1+\beta}2^{nk}\left\Vert \phi_{n}\right\Vert _{l}^{2},
\]
where $\Phi$ is given in (\ref{Eq_4}), and $\left\Vert \cdot\right\Vert _{l}$
denotes the norm on $\mathcal{H}_{l}^{\widehat{\otimes}n}$.

We now define%
\[
\mathcal{H}_{l,k,\beta}=\left\{  \Phi\left(  W\right)  =\sum\limits_{n=0}%
^{\infty}\left\langle :W^{\otimes n}:,\phi_{n}\right\rangle \in\left(
L_{\mathbb{C}}^{2}\right)  ;\left\Vert \Phi\right\Vert _{l,k,\beta}^{2}%
<\infty\right\}  .
\]
The space $\mathcal{H}_{l,k,\beta}$ is a Hilbert space with inner product%
\[
\left\langle \Phi,\Psi\right\rangle _{l,k,\beta}=\sum\limits_{n=0}^{\infty
}\left(  n!\right)  ^{1+\beta}2^{nk}\left\langle \phi_{n},\psi_{n}%
\right\rangle _{l},
\]
where $\Phi$, $\Psi\in\left(  L_{\mathbb{C}}^{2}\right)  $ are as in
(\ref{Eq_4})-(\ref{Eq_5}), and $\left\langle \cdot,\cdot\right\rangle _{l}$
denotes the inner product on $\mathcal{H}_{l}^{\widehat{\otimes}n}$.

The Kondratiev space of test functions $\left(  \mathcal{H}_{\infty}\right)
^{\beta}$ is defined to be the projective limit of the spaces $\mathcal{H}%
_{l,k,\beta}$:%
\[
\left(  \mathcal{H}_{\infty}\right)  ^{\beta}=\underleftarrow{\lim}%
_{l,k\in\mathbb{N}}\mathcal{H}_{l,k,\beta}.
\]
As a vector space $\left(  \mathcal{H}_{\infty}\right)  ^{\beta}=\cap
_{l,k\in\mathbb{N}}\mathcal{H}_{l,k,\beta}$. The space of test functions
$\left(  \mathcal{H}_{\infty}\right)  ^{\beta}$ is a nuclear countable Hilbert
space, which is continuously and densely embedded in $\left(  L_{\mathbb{C}%
}^{2}\right)  $. Moreover, $\left(  \mathcal{H}_{\infty}\right)  ^{\beta}$ and
its topology do not depend on the family of Hilbertian norms $\left\{
\left\Vert \cdot\right\Vert _{l}\right\}  _{l\in\mathbb{N}}$, see e.g.
\cite[Theorem 1]{KLS96}, \cite[Chapter IV, Theorem 1.4]{Huang-Yang}.

The construction used to obtain the spaces $\left(  \mathcal{H}_{\infty
}\right)  ^{\beta}$ can be carried out starting with an arbitrary nuclear
space $\mathcal{N}$. For $0\leq\beta\leq1$, the spaces $\left(  \mathcal{N}%
\right)  ^{\beta}$ were studied by Kondratiev, Leukert and Streit in
\cite{KphD}, \cite{KS93}, \cite{KLS96}, see also \cite[Chapter IV]%
{Huang-Yang}. In the case $\beta=0$ and $\mathcal{N=S}$, the Schwartz space in
$\mathbb{R}^{n}$, the space $\left(  \mathcal{N}\right)  ^{0}$ is the Hida
space of test functions, see e.g. \cite{Hida et al}.

\subsection{Kondratiev-type spaces of distributions}

Let $\mathcal{H}_{-l,-k,-\beta}$ be the dual with respect to $(L_{\mathbb{C}%
}^{2})$ of $\mathcal{H}_{l,k,\beta}$ and let \ $(\mathcal{H}_{\infty}%
)^{-\beta}$ be the dual with respect to\ $(L_{\mathbb{C}}^{2})$ of
$(\mathcal{H}_{\infty})^{\beta}$. We denote by $\left\langle \left\langle
\cdot,\cdot\right\rangle \right\rangle $\ the corresponding dual pairing which
is given by the extension of the scalar product on $(L_{\mathbb{C}}^{2})$. We
define the expectation of a distribution $\boldsymbol{\Phi}\in(\mathcal{H}%
_{\infty})^{-\beta}$ as $\mathbb{E}_{\mu}(\boldsymbol{\Phi})=\left\langle
\left\langle \boldsymbol{\Phi},1\right\rangle \right\rangle $.

The dual space of \ $(\mathcal{H}_{\infty})^{-\beta}$ is given by%
\[
(\mathcal{H}_{\infty})^{-\beta}=\underset{l,k\in\mathbb{N}}{{\LARGE \cup}%
}\mathcal{H}_{-l,-k,-\beta},
\]
see \cite[Chapter IV, Theorem 1.5]{Huang-Yang}. We will consider
$(\mathcal{H}_{\infty})^{-\beta}$ with the inductive limit topology. In
particular, we know that every distribution is of finite order, i.e. for any
$\boldsymbol{\Phi}\in(\mathcal{H}_{\infty})^{-\beta}$ there exist
$l,k\in\mathbb{N}$ such that $\boldsymbol{\Phi}\in\mathcal{H}_{-l,-k,-\beta}$.
The chaos decomposition introduces a natural decomposition of
$\boldsymbol{\Phi}\in(\mathcal{H}_{\infty})^{-\beta}$ into generalized kernels
$\Phi_{n}\in(\mathcal{H}_{\infty}^{\ast}\mathbb{(C}))^{\widehat{\otimes}n}$.
Let $\Phi_{n}\in(\mathcal{H}_{\infty}^{\ast}\mathbb{(C}))^{\widehat{\otimes}%
n}$ be given. Then there is a distribution, denoted as $\left\langle \Phi
_{n},:W^{\otimes n}:\right\rangle $, in $(\mathcal{H}_{\infty})^{-\beta}$
acting on $\Psi\in$\ $(\mathcal{H}_{\infty})^{\beta}$ ($\Psi=\sum
\limits_{n=0}^{\infty}\left\langle :\cdot^{\otimes n}:,\psi_{n}\right\rangle
,$ with $\psi_{n}\in\left(  \mathcal{H}_{\infty}(\mathbb{C})\right)
^{\widehat{\otimes}n}$)\ as
\[
\left\langle \left\langle \left\langle \Phi_{n},:W^{\otimes n}:\right\rangle
,\Psi\right\rangle \right\rangle =n!\left\langle \Phi_{n},\psi_{n}%
\right\rangle .
\]

Any $\boldsymbol{\Phi}\in(\mathcal{H}_{\infty})^{-\beta}$ has a unique
decomposition of the form%
\[
\boldsymbol{\Phi}=\overset{\infty}{\underset{n=0}{\sum}}\left\langle \Phi
_{n},:W^{\otimes n}:\right\rangle \text{, }\Phi_{n}\in(\mathcal{H}_{\infty
}^{\ast}\mathbb{(C}))^{\widehat{\otimes}n}\text{,}%
\]
where the series converges in $(\mathcal{H}_{\infty})^{-\beta}$, in addition,
we have
\[
\left\langle \left\langle \boldsymbol{\Phi},\Psi\right\rangle \right\rangle
=\overset{\infty}{\underset{n=0}{\sum}}n!\left\langle \Phi_{n},\psi
_{n}\right\rangle \text{,\ }\Psi\in(\mathcal{H}_{\infty})^{\beta}.
\]

Now, $\mathcal{H}_{-l,-k,-\beta}$\ is a Hilbert space, that can be described
as follows:%
\[
\mathcal{H}_{-l,-k,-\beta}=\left\{  \boldsymbol{\Phi}\in(\mathcal{H}_{\infty
})^{-\beta};\text{ }\left\Vert \boldsymbol{\Phi}\right\Vert _{-l,-k,-\beta
}<\infty\right\}  ,
\]
where
\begin{equation}
\left\Vert \boldsymbol{\Phi}\right\Vert _{-l,-k,-\beta}^{2}=\underset
{n=0}{\overset{\infty}{\sum}}\left(  n!\right)  ^{1-\beta}2^{-nk}\left\Vert
\Phi_{n}\right\Vert _{-l}^{2}, \label{Eq_6}%
\end{equation}
see \ \cite[Chapter IV, Theorem 1.5]{Huang-Yang}.

\begin{remark}
Notice that%
\begin{align*}
(\mathcal{H}_{\infty})^{1}  &  \subset\cdots\subset(\mathcal{H}_{\infty
})^{\beta}\subset\cdots\subset(\mathcal{H}_{\infty})^{0}\subset(L_{\mathbb{C}%
}^{2})\\
&  \subset(\mathcal{H}_{\infty})^{-0}\subset\cdots\subset(\mathcal{H}_{\infty
})^{-\beta}\subset\cdots\subset(\mathcal{H}_{\infty})^{-1}.
\end{align*}
Following Kondratiev, Leukert and Streit, in this article we work with the
Gel'fand triple $(\mathcal{H}_{\infty})^{1}\subset(L_{\mathbb{C}}^{2}%
)\subset(\mathcal{H}_{\infty})^{-1}$.
\end{remark}

\subsection{The $S$-transform and the characterization of $(\mathcal{H}%
_{\infty})^{-1}$}

\subsubsection{The $S$-transform}

We first consider the Wick exponential:
\[
:\exp\left\langle W,g\right\rangle :=\exp\left(  \left\langle W,g\right\rangle
-\frac{1}{2}\left\Vert g\right\Vert _{0}^{2}\right)  =\overset{\infty
}{\underset{n=0}{\sum}}\frac{1}{n!}\left\langle :W^{\otimes n}:,g^{\otimes
n}\right\rangle \text{, }%
\]
\ for $W\in\mathcal{H}_{\infty}^{\ast}\left(  \mathbb{R}\right)  $,
$g\in\mathcal{H}_{\infty}\left(  \mathbb{C}\right)  $. Then $:\exp\left\langle
W,g\right\rangle :\in(L_{\mathbb{C}}^{2})$ and its $l$, $k$, $1$-norm is given
by%
\[
\left\Vert :\exp\left\langle \cdot,g\right\rangle :\right\Vert _{l,k,1}%
^{2}=\underset{n=0}{\overset{\infty}{\sum}}(n!)^{2}2^{nk}\left\Vert \frac
{1}{n!}g^{\otimes n}\right\Vert _{l}^{2}=\underset{n=0}{\overset{\infty}{\sum
}}\left(  2^{k}\left\Vert g\right\Vert _{l}^{2}\right)  ^{n}.
\]
This norm is finite if and only if $2^{k}\left\Vert g\right\Vert _{l}^{2}<1$,
i.e. $:\exp\left\langle W,g\right\rangle :\in\mathcal{H}_{l,k,\beta}$ if and
only if $g$ belongs to the following neighborhood of zero:%
\[
\mathcal{U}_{l,k}=\left\{  f\in\mathcal{H}_{\infty}\left(  \mathbb{C}\right)
;\left\Vert f\right\Vert _{l}<\frac{1}{2^{\frac{k}{2}}}\right\}  .
\]

Therefore the Wick exponential does not belong to $(\mathcal{H}_{\infty})^{1}%
$, i.e. it is not a test function, in contrast to usual white noise analysis.

Let $\boldsymbol{\Phi}\in(\mathcal{H}_{\infty})^{-1},$ then there \ exist
$l,k$ such that $\boldsymbol{\Phi}\in$ $\mathcal{H}_{-l,-k,-1}$. For all
$f\in\mathcal{U}_{l,k}$, we define the (local) $S$-transform of
$\boldsymbol{\Phi}$ as
\begin{equation}
{\LARGE S}\boldsymbol{\Phi}\left(  f\right)  =\left\langle \left\langle
\boldsymbol{\Phi},:\exp\left\langle \cdot,f\right\rangle :\right\rangle
\right\rangle =\underset{n=0}{\overset{\infty}{\sum}}\left\langle \Phi
_{n},f^{\otimes n}\right\rangle . \label{Eq_7}%
\end{equation}
Hence, for $\boldsymbol{\Phi}\in$ $\mathcal{H}_{-l,-k,-1}$, (\ref{Eq_7})
defines the $S$-transform for all $f$ $\in\mathcal{U}_{l,k}$.

\subsubsection{Holomorphic functions on $\mathcal{H}_{\infty}(\mathbb{C})$}

Let $\mathcal{V}_{l,\epsilon}=\left\{  f\in\mathcal{H}_{\infty}\left(
\mathbb{C}\right)  ;\left\Vert f\right\Vert _{l}<\epsilon\right\}  $ be a
neighborhood of zero in $\mathcal{H}_{\infty}\left(  \mathbb{C}\right)  $. A
map $F:\mathcal{V}_{l,\epsilon}\rightarrow\mathbb{C}$ is called
\textit{holomorphic} in $\mathcal{V}_{l,\epsilon}$, if it satisfies the
following two conditions: (i) for each $g_{0}\in\mathcal{V}_{l,\epsilon}$,
$g\in\mathcal{H}_{\infty}\left(  \mathbb{C}\right)  $ there exists a
neighborhood $V_{g_{0},g}$ in $\mathbb{C}$ around the origin such that the map
$z\rightarrow F\left(  g_{0}+zg\right)  $ is holomorphic in $V_{g_{0},g}$.
(ii) For each $g\in\mathcal{V}_{l,\epsilon}$ there exists an open set
$\mathcal{U}\subset\mathcal{V}_{l,\epsilon}$ containing $g$ such that
$F\left(  \mathcal{U}\right)  $\ is bounded.

By identifying two maps $F_{1}$ and $F_{2}$ coinciding in a neighborhood of
zero, we define $Hol_{0}(\mathcal{H}_{\infty}(\mathbb{C}))$ as the space of
germs of holomorphic maps around the origin.

\subsubsection{\label{Section_Characterization}Characterization of
$(\mathcal{H}_{\infty})^{-1}$}

A key result is the following: the mapping
\[%
\begin{array}
[c]{cccc}%
{\LARGE S}: & (\mathcal{H}_{\infty})^{-1} & \rightarrow & Hol_{0}%
(\mathcal{H}_{\infty}(\mathbb{C}))\\
& \boldsymbol{\Phi} & \rightarrow & S\boldsymbol{\Phi}%
\end{array}
\]
is a well-defined bijection, see \cite[Theorem 3]{KLS96}, \ \cite[Chapter IV,
Theorem 2.13]{Huang-Yang}.

\subsubsection{Integration of distributions}

Let $\left(  \mathfrak{L},\mathcal{A},\nu\right)  $ be a measure space, and
\[%
\begin{array}
[c]{ccc}%
\mathfrak{L} & \rightarrow & (\mathcal{H}_{\infty})^{-1}\\
\mathfrak{l} & \rightarrow & \boldsymbol{\Phi}_{\mathfrak{l}}%
\end{array}
.
\]
Assume that there exists an open neighborhood $\mathcal{V}\subset
\mathcal{H}_{\infty}(\mathbb{C})$ of zero such that (i) $S\boldsymbol{\Phi
}_{\mathfrak{l}}$, $\mathfrak{l}\in\mathfrak{L}$, is holomorphic in
$\mathcal{V}$; (ii) the mapping $\mathfrak{l}\rightarrow{\LARGE S}%
\boldsymbol{\Phi}_{\mathfrak{l}}\left(  g\right)  $ is measurable for every
$g\in\mathcal{V}$; and (iii) there exists a function $C(\mathfrak{l})\in
L^{1}\left(  \mathfrak{L},\mathcal{A},\nu\right)  $ such that $\left\vert
{\LARGE S}\boldsymbol{\Phi}_{\mathfrak{l}}\left(  g\right)  \right\vert \leq
C(\mathfrak{l})$ for all $g\in\mathcal{V}$ and for $\nu$-almost $\mathfrak{l}%
\in\mathfrak{L}$. Then there exist $l_{0}$, $k_{0}$ $\in\mathbb{N}$ such that
$\int_{\mathfrak{L}}\boldsymbol{\Phi}_{\mathfrak{l}}d\nu\left(  \mathfrak{l}%
\right)  $ exists as a Bochner integral in $\mathcal{H}_{-l_{0},-k_{0},-1}$,
in particular,
\begin{equation}
{\LARGE S}\left(  \int\nolimits_{\mathfrak{L}}\boldsymbol{\Phi}_{\mathfrak{l}%
}d\nu\left(  \mathfrak{l}\right)  \right)  \left(  g\right)  =\int
\nolimits_{\mathfrak{L}}{\LARGE S}\boldsymbol{\Phi}_{\mathfrak{l}}\left(
g\right)  d\nu\left(  \mathfrak{l}\right)  \text{, for any }g\in
\mathcal{V}\text{,} \label{Eq_C}%
\end{equation}
cf. \cite[Theorem 6]{KLS96}, \cite[Chapter IV, Theorem 2.15]{Huang-Yang}.

\subsubsection{The Wick product}

Given $\boldsymbol{\Phi}$, $\boldsymbol{\Psi}\in(\mathcal{H}_{\infty})^{-1}$,
we define the \textit{Wick product} of them as
\[
\boldsymbol{\Phi}\Diamond\boldsymbol{\Psi}=S^{-1}\left(  {\LARGE S}%
\boldsymbol{\Phi}{\LARGE S}\boldsymbol{\Psi}\right)  .
\]
This product is well-defined because $Hol_{0}(\mathcal{H}_{\infty}%
(\mathbb{C}))$\ is an algebra. The map%
\[%
\begin{array}
[c]{ccc}%
(\mathcal{H}_{\infty})^{-1}\times(\mathcal{H}_{\infty})^{-1} & \rightarrow &
(\mathcal{H}_{\infty})^{-1}\\
\left(  \boldsymbol{\Phi},\boldsymbol{\Psi}\right)  & \rightarrow &
\boldsymbol{\Phi}\Diamond\boldsymbol{\Psi}%
\end{array}
\]
is well-defined and continuous. Furthermore, if $\boldsymbol{\Phi}%
\in\mathcal{H}_{-l_{1},-k_{1},-1}$, $\boldsymbol{\Psi}\in\mathcal{H}%
_{-l_{2},-k_{2},-1}$, and $l:=\max\left\{  l_{1},l_{2}\right\}  $,
$k:=k_{1}+k_{2}+1$, then
\[
\left\Vert \boldsymbol{\Phi}\Diamond\boldsymbol{\Psi}\right\Vert
_{-l,-k,-1}\leq\left\Vert \boldsymbol{\Phi}\right\Vert _{-l_{1},-k_{1}%
,-1}\left\Vert \boldsymbol{\Psi}\right\Vert _{-l_{2},-k_{2},-1},
\]
cf. \cite[Proposition 11]{KLS96}. The Wick product leaves $(\mathcal{H}%
_{\infty})$ invariant. By induction on $n$, we can define the Wick powers:
\[
\boldsymbol{\Phi}^{\Diamond n}={\LARGE S}^{-1}(\left(  {\LARGE S}%
\boldsymbol{\Phi}\right)  ^{n})\in(\mathcal{H}_{\infty})^{-1}.
\]
Consequently $\sum_{n=0}^{m}a_{n}\boldsymbol{\Phi}^{\Diamond n}\in
(\mathcal{H}_{\infty})^{-1}$.

\subsubsection{\label{Section Wick_analytic_fun}Wick analytic functions in
$(\mathcal{H}_{\infty})^{-1}$}

Assume that $F$ is an analytic function in a neighborhood of \ the point
$z_{0}=\mathbb{E}_{\mu}\left(  \boldsymbol{\Phi}\right)  $ in $\mathbb{C}$,
with $\boldsymbol{\Phi}\in(\mathcal{H}_{\infty})^{-1}$. Then $F^{\Diamond
}(\boldsymbol{\Phi})=S^{-1}(F(S\boldsymbol{\Phi}))$ exists in $(\mathcal{H}%
_{\infty})^{-1}$, cf. \cite[Theorem 12]{KLS96}. \ In addition, if $F$ is
analytic in $z_{0}=\mathbb{E}_{\mu}\left(  \boldsymbol{\Phi}\right)  $, with
power series $F(z)=\sum_{n=0}^{\infty}c_{n}\left(  z-z_{0}\right)  ^{n}$, then
the Wick series $\sum_{n=0}^{\infty}c_{n}\left(  \boldsymbol{\Phi}%
-z_{0}\right)  ^{\Diamond n}$ converges in $(\mathcal{H}_{\infty})^{-1}$ and
$F^{\Diamond}(\boldsymbol{\Phi})=\sum_{n=0}^{\infty}c_{n}\left(
\boldsymbol{\Phi}-z_{0}\right)  ^{\Diamond n}$.

\section{Schwinger Functions and Euclidean Quantum Field Theory}

\subsection{Schwinger functions}

\begin{definition}
Let $f_{1},\ldots,f_{n}\in\mathcal{H}_{\infty}\left(  \mathbb{R}\right)  $,
$n\in\mathbb{N}$. The $n$-th Schwinger function corresponding to
$\boldsymbol{\Phi}\in(\mathcal{H}_{\infty})^{-1}$, with $\mathbb{E}_{\mu
}\left(  \boldsymbol{\Phi}\right)  =1$, is defined as%
\begin{equation}
\mathcal{S}_{n}^{\boldsymbol{\Phi}}\left(  f_{1}\otimes\cdots\otimes
f_{n}\right)  \left(  W\right)  =\left\{
\begin{array}
[c]{lll}%
1 & \text{if} & n=0\\
&  & \\
\left\langle \left\langle \boldsymbol{\Phi},\left\langle W,f_{1}\right\rangle
\cdots\left\langle W,f_{n}\right\rangle \right\rangle \right\rangle  &
\text{if} & n\geq1,
\end{array}
\right.  \label{Eq_8}%
\end{equation}
for $W\in\mathcal{H}_{\infty}^{\ast}\left(  \mathbb{R}\right)  $.
\end{definition}

The pairing in (\ref{Eq_8}) is well-defined because the Wick polynomials
$\mathcal{P}\left(  \mathcal{H}_{\infty}^{\ast}\left(  \mathbb{R}\right)
\right)  $ are dense in $\left(  \mathcal{H}_{\infty}\right)  ^{1}$.

The ${\LARGE T}$-transform of a distribution is defined as
\begin{equation}
{\LARGE T}\boldsymbol{\Phi}\left(  g\right)  =\exp\left(  \frac{-1}%
{2}\left\Vert g\right\Vert _{0}^{2}\right)  S\boldsymbol{\Phi}\left(
ig\right)  \label{Eq_T}%
\end{equation}
for $\boldsymbol{\Phi}\in(\mathcal{H}_{\infty})^{-1}$ and $g\in\mathcal{U}$,
where $\mathcal{U}$ is neighborhood of zero in $\mathcal{H}_{\infty}\left(
\mathbb{C}\right)  $. The Schwinger functions can be computed by using the $T$-transform:

\begin{lemma}
[{\cite[Proposition III.3]{GS1999}}]\label{Lemma0}Let $f_{1},\ldots,f_{n}%
\in\mathcal{H}_{\infty}\left(  \mathbb{R}\right)  $, $n\in\mathbb{N}$. The
$n$-th Schwinger function corresponding to $\boldsymbol{\Phi}\in
(\mathcal{H}_{\infty})^{-1}$ is given by%
\[
\mathcal{S}_{n}^{\boldsymbol{\Phi}}\left(  f_{1}\otimes\cdots\otimes
f_{n}\right)  =(-i)^{n}\frac{\partial^{n}}{\partial t_{1}\cdots\partial t_{n}%
}{\LARGE T}\boldsymbol{\Phi}\left(  t_{1}f_{1}+\cdots+t_{n}f_{n}\right)
{\LARGE \mid}_{t_{1}=\cdots=t_{n}=0}.
\]

\end{lemma}

\begin{lemma}
\label{Lemma1}For each distribution $\boldsymbol{\Phi}\in(\mathcal{H}_{\infty
})^{-1}$, with $\mathbb{E}_{\mu}\left(  \boldsymbol{\Phi}\right)  =1 $, the
Schwinger functions $\left\{  \mathcal{S}_{n}^{\boldsymbol{\Phi}}\right\}
_{n\in\mathbb{N}}$\ satisfy the following conditions:

\begin{enumerate}
\item[(OS1)] the sequence $\left\{  \mathcal{S}_{n}^{\boldsymbol{\Phi}%
}\right\}  _{n\in\mathbb{N}}$, with $\mathcal{S}_{n}^{\boldsymbol{\Phi}}%
\in\left(  \mathcal{H}_{\infty}^{\ast}\left(  \mathbb{C}\right)  \right)
^{\otimes n}$, satisfies%
\[
\left\vert \mathcal{S}_{n}^{\boldsymbol{\Phi}}\left(  f_{1}\otimes
\cdots\otimes f_{n}\right)  \right\vert \leq KC^{n}n!%
{\textstyle\prod\nolimits_{i=1}^{n}}
\left\Vert f_{i}\right\Vert _{l},
\]
for some $l$, $k$ $\in\mathbb{N}$, where $K=\sqrt{I_{0}(2^{-k})}\left\Vert
\Phi\right\Vert _{-l.-k-1}$, here $I_{0}$ is the modified Bessel function of
order zero, which satisfies $I_{0}(2^{-k})<1.3$, $C=e2^{\frac{k}{2}}$, and for
any $f_{1},\cdots,f_{n}\in\mathcal{H}_{\infty}\left(  \mathbb{R}\right)  $;

\item[(OS4)] for $n\geq2$ and all $\sigma\in\mathfrak{S}_{n}$, the permutation
group of order $n$, it holds that
\[
\mathcal{S}_{n}^{\boldsymbol{\Phi}}\left(  f_{1}\otimes\cdots\otimes
f_{n}\right)  =\mathcal{S}_{n}^{\boldsymbol{\Phi}}\left(  f_{\sigma\left(
1\right)  }\otimes\cdots\otimes f_{\sigma\left(  n\right)  }\right)  ,
\]
for any $f_{1},\cdots,f_{n}\in\mathcal{H}_{\infty}\left(  \mathbb{R}\right)  $.
\end{enumerate}
\end{lemma}

\begin{proof}
Estimation (OS1) is given in the proof of Theorem 2 in \cite{KSW95}. The
Schwinger functions $\left(  \mathcal{S}_{n}^{\boldsymbol{\Phi}}\right)  $ are
symmetric by definition.
\end{proof}

\subsection{\label{Section_white_noise_process}A white-noise process}

For $t\in\mathbb{Q}_{p}$, $\overrightarrow{x}\in\mathbb{Q}_{p}^{N-1}$, we set
$x=\left(  t,\overrightarrow{x}\right)  $. We denote by $\delta_{x}%
:=\delta_{\left(  t,\overrightarrow{x}\right)  }$, the Dirac distribution at
$\left(  t,\overrightarrow{x}\right)  $.

\begin{lemma}
$\delta_{\left(  t,\overrightarrow{x}\right)  }\in\left(  \mathcal{H}_{\infty
}\right)  ^{-1}$.
\end{lemma}

\begin{proof}
We first notice that%
\[
\left\Vert \delta_{\left(  t,\overrightarrow{x}\right)  }\right\Vert _{-l}%
^{2}=\int\nolimits_{\mathbb{Q}_{p}^{N}}\frac{d^{N}\xi}{\left[  \xi\right]
_{p}^{l}}<\infty\text{ for }l>N\text{,}%
\]
which implies that $\delta_{\left(  t,\overrightarrow{x}\right)  }%
\in\mathcal{H}_{-l}(\mathbb{C})$ for all $l>N$, see (\ref{Eq_A}). Now, we
define $\left\{  \Phi_{n}\right\}  _{n\in\mathbb{N}}$, with $\Phi_{n}%
\in\left(  \mathcal{H}_{\infty}^{\ast}\left(  \mathbb{C}\right)  \right)
^{\widehat{\otimes}n}$, as $\Phi_{n}=0$ if $n\neq1$ and $\Phi_{1}%
=\delta_{\left(  t,\overrightarrow{x}\right)  }$. Then
\[
\sum_{n}\left\langle \Phi_{n},:W^{\otimes n}:\right\rangle =\left\langle
\delta_{\left(  t,\overrightarrow{x}\right)  },:W:\right\rangle \in\left(
\mathcal{H}_{\infty}\right)  ^{-1}.
\]
In addition, for $\psi\in\mathcal{H}_{\infty}\left(  \mathbb{C}\right)  $, we
have
\begin{align*}
\left\langle \left\langle \left\langle \delta_{\left(  t,\overrightarrow
{x}\right)  },:W:\right\rangle ,\psi\right\rangle \right\rangle  &
=\left\langle \delta_{\left(  t,\overrightarrow{x}\right)  },\psi\right\rangle
=\int\nolimits_{\mathbb{Q}_{p}^{N}}\chi_{p}\left(  -\xi\cdot x\right)
\widehat{\psi}\left(  x\right)  d^{N}\xi\\
&  =\psi\left(  t,\overrightarrow{x}\right)  ,
\end{align*}
where we used that $\psi$ is a continuous function in $L^{1}\cap L^{2}$, see
Section \ref{Section1}\ and \cite[Theorem 10.15]{KKZuniga}.
\end{proof}

We now set%
\[
\boldsymbol{\Phi}\left(  t,\overrightarrow{x}\right)  :=\left\langle
\delta_{\left(  t,\overrightarrow{x}\right)  },:W:\right\rangle \in\left(
\mathcal{H}_{\infty}\right)  ^{-1}.
\]
Then $\boldsymbol{\Phi}\left(  t,\overrightarrow{x}\right)  $\ is a
white-noise process with $\mathbb{E}_{\mu}\left(  \boldsymbol{\Phi}\left(
t,\overrightarrow{x}\right)  \right)  =0$.

Assume that%
\[
H(z)=\sum\nolimits_{k=0}^{\infty}\frac{1}{k!}H_{k}z^{k}\text{, }z\in
U\subset\mathbb{C}\text{,}%
\]
is a holomorphic function \ in $U$, an open neighborhood of $0=\mathbb{E}%
_{\mu}\left(  \boldsymbol{\Phi}\left(  t,\overrightarrow{x}\right)  \right)
$. By \cite[Theorem 12]{KLS96}, see also Section
\ref{Section Wick_analytic_fun}, we can define%
\begin{align*}
H^{\lozenge}(\boldsymbol{\Phi}\left(  t,\overrightarrow{x}\right)  )  &
=\sum\nolimits_{k=0}^{\infty}\frac{1}{k!}H_{k}\boldsymbol{\Phi}\left(
t,\overrightarrow{x}\right)  ^{\lozenge k}\\
&  =\sum\nolimits_{k=0}^{\infty}\frac{1}{k!}H_{k}\left\langle \delta_{\left(
t,\overrightarrow{x}\right)  }^{\otimes k},:W^{\otimes k}:\right\rangle
\in\left(  \mathcal{H}_{\infty}\right)  ^{-1}.
\end{align*}
Our next goal is the construction of the potential%
\begin{equation}
\int\nolimits_{\mathbb{Q}_{p}^{N}}H^{\lozenge}(\boldsymbol{\Phi}\left(
x\right)  )d^{N}x \label{Eq_12}%
\end{equation}
as a white-noise distribution. This goal is accomplished through the following result:

\begin{theorem}
\label{Theorem1}(i) Let $H$ be a holomorphic function at zero such that
$H(0)=0$. Then (\ref{Eq_12}) exists as a Bochner integral in a suitable
subspace of $\left(  \mathcal{H}_{\infty}\right)  ^{-1}$.

\noindent(ii) The distribution%
\[
\boldsymbol{\Phi}_{H}:=\exp^{\lozenge}\left(  -\int\nolimits_{\mathbb{Q}%
_{p}^{N}}H^{\lozenge}(\boldsymbol{\Phi}\left(  x\right)  )d^{N}x\right)
\]
is an element of $\left(  \mathcal{H}_{\infty}\right)  ^{-1}$.

\noindent(iii) The ${\LARGE T}$-transform of $\boldsymbol{\Phi}_{H}$\ is given
by
\[
{\LARGE T}\boldsymbol{\Phi}_{H}\left(  g\right)  =\exp\left(  -\int
\nolimits_{\mathbb{Q}_{p}^{N}}H(ig\left(  x\right)  )+\frac{1}{2}\left(
g\left(  x\right)  \right)  ^{2}\text{ }d^{N}x\right)
\]
for all $g$ in a neighborhood $\mathcal{U\subset H}_{\infty}\mathcal{(}%
\mathbb{C}\mathcal{)}$ of the zero. In particular, $\mathbb{E}_{\mathbb{\mu}%
}(\boldsymbol{\Phi}_{H})=1$.
\end{theorem}

\begin{proof}
(i) The result follows from the discussion presented in Section
\ref{Section_Characterization}, see also \cite[Theorem 6]{KLS96}, as follows.
Let $r>0$ be the radius of convergence of the Taylor series of $H$ at the
origin. We set $C(N):=\sqrt{\int_{\mathbb{Q}_{p}^{N}}\frac{d^{N}\xi}{\left[
\xi\right]  _{p}^{l}}}$, for a fixed $l>N$, and%
\[
\mathcal{U}_{0}:=\left\{  g\in\mathcal{H}_{\infty}\left(  \mathbb{C}\right)
;\left\Vert g\right\Vert _{l}<\frac{r}{C(N)}\right\}  .
\]
Then, for $g\in\mathcal{U}_{0}$ we have%
\begin{align}
{\LARGE S}H^{\lozenge}(\boldsymbol{\Phi}\left(  x\right)  )\left(  g\right)
&  =\sum_{k=1}^{\infty}\frac{1}{k!}H_{k}\left\langle \delta_{x}^{\otimes
k},g^{\otimes k}\right\rangle =\sum_{k=1}^{\infty}\frac{1}{k!}H_{k}%
g(x)^{k}\label{Eq_13}\\
&  =\sum_{k=1}^{\infty}\frac{1}{k!}H_{k}\left\{  \frac{g(x)}{r}\right\}
^{k}r^{k}.\nonumber
\end{align}
By Claim A, $\left\vert \frac{g(x)}{r}\right\vert <1$, and from (\ref{Eq_13})
we obtain that%
\begin{equation}
\left\vert {\LARGE S}H^{\lozenge}(\boldsymbol{\Phi}\left(  x\right)  )\left(
g\right)  \right\vert \leq\left\vert g(x)\right\vert \sum_{k=1}^{\infty}%
\frac{1}{k!}\left\vert H_{k}\right\vert r^{k-1}\in L^{1}\left(  \mathbb{Q}%
_{p}^{N}\right)  , \label{Eq_14}%
\end{equation}
because $\mathcal{H}_{\infty}\left(  \mathbb{C}\right)  \subset L^{1}\left(
\mathbb{Q}_{p}^{N}\right)  $, cf. \cite[Theorem 10.15]{KKZuniga}. Estimation
(\ref{Eq_14}) implies the holomorphy of ${\LARGE S}H^{\lozenge}%
(\boldsymbol{\Phi}\left(  x\right)  )\left(  g\right)  $ for any
$g\in\mathcal{U}_{0}$. Since ${\LARGE S}H^{\lozenge}(\boldsymbol{\Phi}\left(
x\right)  )\left(  g\right)  $ is measurable by \cite[Theorem 6]{KLS96}, we
conclude that (\ref{Eq_12}) is an element of $\left(  \mathcal{H}_{\infty
}\right)  ^{-1}$.

\textbf{Claim A.} $\ \mathcal{U}_{0}\subset\mathcal{U}:=\left\{
g\in\mathcal{H}_{\infty}\left(  \mathbb{C}\right)  ;\left\Vert g\right\Vert
_{L^{\infty}}<r\right\}  .$

The Claim follows from the fact that
\[
\left\Vert g\right\Vert _{L^{\infty}}\leq C(N)\left\Vert g\right\Vert
_{l}\text{, for\ }g\in\mathcal{H}_{\infty}\left(  \mathbb{C}\right)  .
\]
\ This last fact is verified as follows: by using that $g\in L^{1}\left(
\mathbb{Q}_{p}^{N}\right)  \cap L^{2}\left(  \mathbb{Q}_{p}^{N}\right)  $, and
the Cauchy-Schwarz inequality, we have%
\begin{align*}
\left\vert g\left(  x\right)  \right\vert  &  =\left\vert \int
\nolimits_{\mathbb{Q}_{p}^{N}}\chi_{p}\left(  -\xi\cdot x\right)  \widehat
{g}\left(  \xi\right)  d^{N}\xi\right\vert \leq\int\nolimits_{\mathbb{Q}%
_{p}^{N}}\left\vert \widehat{g}\left(  \xi\right)  \right\vert d^{N}\xi\\
&  =\int\nolimits_{\mathbb{Q}_{p}^{N}}\frac{1}{\left[  \xi\right]  _{p}%
^{\frac{l}{2}}}\text{ }\left\{  \left[  \xi\right]  _{p}^{\frac{l}{2}%
}\left\vert \widehat{g}\left(  \xi\right)  \right\vert \right\}  d^{N}\xi\leq
C(N)\left\Vert g\right\Vert _{l}.
\end{align*}

(ii) Since $\exp$ is analytic in a neighborhood of $0=\mathbb{E}_{\mu}\left(
\boldsymbol{\Phi}\left(  t,\overrightarrow{x}\right)  \right)  $, then
\[
\exp^{\lozenge}\left(  -\int\nolimits_{\mathbb{Q}_{p}^{N}}H^{\lozenge
}(\boldsymbol{\Phi}\left(  x\right)  )d^{N}x\right)  ={\LARGE S}^{-1}\left(
\exp\left(  {\LARGE S}\left(  -\int\nolimits_{\mathbb{Q}_{p}^{N}}H^{\lozenge
}(\boldsymbol{\Phi}\left(  x\right)  )d^{N}x\right)  \right)  \right)  ,
\]
and by (i), $-\int\nolimits_{\mathbb{Q}_{p}^{N}}H^{\lozenge}(\boldsymbol{\Phi
}\left(  x\right)  )d^{N}x\in\left(  \mathcal{H}_{\infty}\right)  ^{-1}$, and
then its $S$-transform is analytic at the origin, and its composition with
$\exp$ gives again an analytic function at the origin, whose inverse
$S$-transform gives an element of $\left(  \mathcal{H}_{\infty}\right)  ^{-1}%
$, cf. \cite[Theorem 12]{KLS96}.

(iii) The calculation of the ${\LARGE T}$-transform uses (\ref{Eq_T}),
$\exp^{\lozenge}\left(  \cdot\right)  ={\LARGE S}^{-1}\left(  \exp\left(
{\LARGE S}\left(  \cdot\right)  \right)  \right)  $, and (\ref{Eq_13}) as
follows:%
\begin{align*}
\left(  {\LARGE T}\boldsymbol{\Phi}_{H}\right)  \left(  g\right)   &
=\exp\left(  -\frac{1}{2}\left\Vert g\right\Vert _{0}^{2}\right)  \exp\left(
{\LARGE S}\left(  -\int_{\mathbb{Q}_{p}^{N}}H^{\lozenge}\left(
\boldsymbol{\Phi}\left(  x\right)  \right)  d^{N}x\right)  \left(  ig\right)
\right) \\
&  =\exp\left(  -\frac{1}{2}\left\Vert g\right\Vert _{0}^{2}\right)
\exp\left(  -\int_{\mathbb{Q}_{p}^{N}}\left\langle \left\langle H^{\lozenge
}\left(  \boldsymbol{\Phi}\left(  x\right)  \right)  ,:\exp\left\langle
\cdot,ig\right\rangle :\right\rangle \right\rangle d^{N}x\right) \\
&  =\exp\left(  -\frac{1}{2}\left\Vert g\right\Vert _{0}^{2}\right)
\exp\left(  -\int_{\mathbb{Q}_{p}^{N}}{\LARGE S}H^{\lozenge}\left(
\boldsymbol{\Phi}\left(  x\right)  \right)  \left(  ig\right)  d^{N}x\right)
\\
&  =\exp\left(  -\int_{\mathbb{Q}_{p}^{N}}H\left(  ig\left(  x\right)
\right)  +\frac{1}{2}g\left(  x\right)  ^{2}\text{ }d^{N}x\right)  .
\end{align*}

In particular $\mathbb{E}_{\mu}(\boldsymbol{\Phi}_{H})={\LARGE T}%
\boldsymbol{\Phi}_{H}\left(  0\right)  =1$.
\end{proof}

\subsection{\label{Sect4}Pseudodifferential Operators and Green Functions}

A non-constant homogeneous polynomial $\mathfrak{l}\left(  \xi\right)
\in\mathbb{Z}_{p}\left[  \xi_{1},\cdots,\xi_{N}\right]  $ of degree $d$ is
called \textit{elliptic\ }if it\textit{ }satisfies $\mathfrak{l}\left(
\xi\right)  =0\Leftrightarrow\xi=0$. There are infi\-nitely many elliptic
polynomials, cf. \cite[Lemma 24]{Zuniga-LNM-2016}. A such polynomial
satisfies
\begin{equation}
C_{0}\left(  \alpha\right)  \left\Vert \xi\right\Vert _{p}^{\alpha d}%
\leq\left\vert \mathfrak{l}\left(  \xi\right)  \right\vert _{p}^{\alpha}\leq
C_{1}\left(  \alpha\right)  \left\Vert \xi\right\Vert _{p}^{\alpha d},
\label{basic_estimate}%
\end{equation}
for some positive constants $C_{0}\left(  \alpha\right)  $, $C_{1}\left(
\alpha\right)  $, cf. \cite[Lemma 25]{Zuniga-LNM-2016}. We define an
\textit{elliptic pseudodifferential operator with symbol }$\left\vert
\mathfrak{l}\left(  \xi\right)  \right\vert _{p}^{\alpha}$, with $\alpha>0$,
as
\begin{equation}
\left(  \mathbf{L}_{\alpha}h\right)  \left(  x\right)  =\mathcal{F}%
_{\xi\rightarrow x}^{-1}\left(  \left\vert \mathfrak{l}\left(  \xi\right)
\right\vert _{p}^{\alpha}\mathcal{F}_{x\rightarrow\xi}h\right)  ,
\label{ellipticoperaator}%
\end{equation}
for $h\in\mathcal{D}(\mathbb{Q}_{p}^{N})$. We define $G:=G\left(
x;m,\alpha\right)  \in\mathcal{D}^{\prime}(\mathbb{Q}_{p}^{N})$, with
$\alpha>0$, $m>0$, to be the solution of
\[
\left(  \boldsymbol{L}_{\alpha}+m^{2}\right)  G=\delta\text{ in \ }%
\mathcal{D}^{\prime}(\mathbb{Q}_{p}^{N}).
\]
We will say that the \textit{Green function} $G\left(  x;m,\alpha\right)  $ is
a \textit{fundamental solution} of the equation%
\begin{equation}
\left(  \boldsymbol{L}_{\alpha}+m^{2}\right)  u=h,\;\text{with }%
h\in\mathcal{D}(\mathbb{Q}_{p}^{N}),\text{ }m>0. \label{equation1}%
\end{equation}
As a distribution \ from $\mathcal{D}^{\prime}(\mathbb{Q}_{p}^{N})$, the Green
function $G\left(  x;m,\alpha\right)  $\ is given by
\begin{equation}
G\left(  x;\alpha,m\right)  =\mathcal{F}_{\xi\rightarrow x}^{-1}\left(
\frac{1}{\left\vert \mathfrak{l}\left(  \xi\right)  \right\vert _{p}^{\alpha
}+m^{2}}\right)  . \label{GreenFunc}%
\end{equation}
Notice that by (\ref{basic_estimate}), we have
\[
\frac{1}{\left\vert \mathfrak{l}\left(  \xi\right)  \right\vert _{p}^{\alpha
}+m^{2}}\in L^{1}\left(  \mathbb{Q}_{p}^{N},d^{N}\xi\right)  \text{\ \ for
\ }\alpha d>N,
\]
and in this case, $G\left(  x;\alpha,m\right)  $ is an $L^{\infty}$-function.

There exists a Green function $G\left(  x;\alpha,m\right)  $ for the operator
$\boldsymbol{L}_{\alpha}+m^{2}$, which is continuous and non-negative on
$\mathbb{Q}_{p}^{n}\smallsetminus\left\{  0\right\}  $, and tends to zero at
infinity. The equation%
\begin{equation}
\left(  \boldsymbol{L}_{\alpha}+m^{2}\right)  u=g\text{, } \label{equation3A}%
\end{equation}
\ with $g\in\mathcal{H}_{\mathbb{\infty}}\left(  \mathbb{R}\right)  $, has a
unique solution $u(x)=G\left(  x;\alpha,m\right)  \ast g(x)\in\mathcal{H}%
_{\mathbb{\infty}}\left(  \mathbb{R}\right)  $, cf. \cite[Theorem
11.2]{KKZuniga}.

As a consequence one obtains that \ the mapping%
\begin{equation}%
\begin{array}
[c]{cccc}%
\mathcal{G}_{\alpha,m}: & \mathcal{H}_{\mathbb{\infty}}\left(  \mathbb{R}%
\right)  & \rightarrow & \mathcal{H}_{\mathbb{\infty}}\left(  \mathbb{R}%
\right) \\
& g\left(  x\right)  & \rightarrow & G\left(  x;\alpha,m\right)  \ast g(x),
\end{array}
\label{Mapping_G}%
\end{equation}
is continuous, cf. \cite[Corollary 11.3]{KKZuniga}.

\begin{remark}
\label{Nota3}For $\alpha>0$, $\beta>0$, $m>0$, we set%
\[
\left(  \mathbf{L}_{\alpha,\beta,m}h\right)  \left(  x\right)  =\mathcal{F}%
_{\xi\rightarrow x}^{-1}\left(  \left(  \left\vert \mathfrak{l}\left(
\xi\right)  \right\vert _{p}^{\alpha}+m^{2}\right)  ^{\beta}\mathcal{F}%
_{x\rightarrow\xi}h\right)  ,
\]
for $h\in\mathcal{D}(\mathbb{Q}_{p}^{N})$. We denote by $G\left(
x;\alpha,\beta,m\right)  $ the associated Green function. By using the fact
that%
\[
C_{0}\left(  \alpha,\beta,m\right)  \left[  \xi\right]  _{p}^{\alpha\beta
d}\leq\left(  \left\vert \mathfrak{l}\left(  \xi\right)  \right\vert
_{p}^{\alpha}+m^{2}\right)  ^{\beta}\leq C_{1}\left(  \alpha,\beta,m\right)
\left[  \xi\right]  _{p}^{\alpha\beta d},
\]
all the results presented in this section for operators $\boldsymbol{L}%
_{\alpha}+m^{2}$ can be extended to operators $\mathbf{L}_{\alpha,\beta,m}$.
In particular,%
\begin{equation}%
\begin{array}
[c]{cccc}%
\mathcal{G}_{\alpha,\beta,m}: & \mathcal{H}_{\mathbb{\infty}}\left(
\mathbb{R}\right)  & \rightarrow & \mathcal{H}_{\mathbb{\infty}}\left(
\mathbb{R}\right) \\
& g\left(  x\right)  & \rightarrow & G\left(  x;\alpha,\beta,m\right)  \ast
g(x),
\end{array}
\label{Mapping_G_2}%
\end{equation}
gives rise to a continuous mapping. \ As operators on $\mathcal{H}%
_{\mathbb{\infty}}\left(  \mathbb{R}\right)  $, we can identify $\mathcal{G}%
_{\alpha,\beta,m}$ with the operator $\left(  \boldsymbol{L}_{\alpha}%
+m^{2}\right)  ^{-\beta}$, which is a pseudodifferential operator with symbol
$\left(  \left\vert \mathfrak{l}\left(  \xi\right)  \right\vert _{p}^{\alpha
}+m^{2}\right)  ^{-\beta}$.
\end{remark}

\begin{remark}
\label{Nota_tensor_2}The mapping
\[%
\begin{array}
[c]{cccc}%
\mathcal{G}_{\alpha,m}^{\otimes2}-1: & \mathcal{H}_{\mathbb{\infty}}%
^{\otimes2} & \rightarrow & \mathcal{H}_{\mathbb{\infty}}^{\otimes2}\\
& f\otimes g & \rightarrow & \mathcal{G}_{\alpha,m}\left(  f\right)
\otimes\mathcal{G}_{\alpha,m}\left(  g\right)  -f\otimes g
\end{array}
\]
is well-defined and continuous. By using \cite[Proposition 1.3.6]{Obata}, any
element $h$ of $\mathcal{H}_{\mathbb{\infty}}^{\otimes2}$ can be represented
as an absolutely convergent series of the form $h=\sum_{i}f_{i}\otimes g_{i}$,
consequently, $\sum_{i}\mathcal{G}_{\alpha,m}\left(  f_{i}\right)
\otimes\mathcal{G}_{\alpha,m}\left(  g_{i}\right)  $ is an element of
$\mathcal{H}_{\mathbb{\infty}}^{\otimes2}$, which implies that $\mathcal{G}%
_{\alpha,m}^{\otimes2}-1$ is a well-defined mapping. On the other hand, the
space $\mathcal{H}_{\mathbb{\infty}}^{\otimes2}$ is locally convex, the
topology is defined by the seminorms%
\[
\left\Vert h\right\Vert _{l,k}=\inf\sum_{i}\left\Vert f_{i}\right\Vert
_{l}\otimes\left\Vert g_{i}\right\Vert _{k}\text{, \ }h\in\mathcal{H}%
_{\mathbb{\infty}}\otimes_{\text{alg}}\mathcal{H}_{\mathbb{\infty}},
\]
where the infimum is taken over all the pairs $\left(  f_{i},g_{j}\right)  $
satisfying $h=\sum_{j}f_{j}\otimes g_{j}$. The continuity of $\mathcal{G}%
_{\alpha,m}^{\otimes2}-1$ is equivalent to%
\[
\left\Vert \left(  \mathcal{G}_{\alpha,m}^{\otimes2}-1\right)  h\right\Vert
_{l,k}\leq C\left\Vert h\right\Vert _{l^{\prime},k^{\prime}},
\]
where the indices $l^{\prime}$, $k^{\prime}$ depend on $l$, $k$. This
condition can be verified easily using the continuity of \ $\mathcal{G}%
_{\alpha,m}$.
\end{remark}

\begin{remark}
\label{Nota_Trace}We denote by $Tr$ (the trace), which is\ the unique element
of $\ \mathcal{H}_{\mathbb{\infty}}^{\ast\widehat{\otimes}2}$ determined by
the formula
\[
\left\langle Tr,f\otimes g\right\rangle =\left\langle f,g\right\rangle
_{0}\text{, for }f,g\in\mathcal{H}_{\mathbb{\infty}}\text{.}%
\]
We define $\left(  \mathcal{G}_{\alpha,m}^{\otimes2}-1\right)  Tr\in
\mathcal{H}_{\mathbb{\infty}}^{\ast\widehat{\otimes}2}$ as
\[
\left\langle \left(  \mathcal{G}_{\alpha,m}^{\otimes2}-1\right)  Tr,f\otimes
g\right\rangle =\left\langle Tr,\left(  \mathcal{G}_{\alpha,m}^{\otimes
2}-1\right)  \left(  f\otimes g\right)  \right\rangle ,
\]
where $\left\langle \cdot,\cdot\right\rangle $ is the pairing between
$\mathcal{H}_{\mathbb{\infty}}^{\ast\widehat{\otimes}2}$ and $\mathcal{H}%
_{\mathbb{\infty}}^{\widehat{\otimes}2}$. For a general construction of this
type of operators the reader may consult \cite[Theorem 9.11]{Kuo}.
\end{remark}

\subsection{L\'{e}vy characteristics}

We recall that an infinitely divisible probability distribution $P$ is a
probability distribution having the property that for each $n\in
\mathbb{N\smallsetminus}\left\{  0\right\}  $ there exists a probability
distribution $P_{n}$ such that $P=P_{n}\ast\cdots\ast P_{n}$ ($n$-times). By
the L\'{e}vy-Khinchine Theorem, see e.g. \cite{Lukacs}, the characteristic
function $C_{P}$ of $P$ satisfies%
\begin{equation}
C_{P}(t)=\int\nolimits_{\mathbb{R}}e^{ist}dP(s)=e^{\digamma\left(  t\right)
}\text{, }t\in\mathbb{R}\text{,} \label{char_function}%
\end{equation}
where $\digamma:\mathbb{R}\rightarrow\mathbb{C}$ is a continuous function,
called the \textit{L\'{e}vy characteristic of} $P$, which is uniquely
represented as follows:%
\[
\digamma\left(  t\right)  =iat-\frac{\sigma^{2}t^{2}}{2}+\int
\nolimits_{\mathbb{R\smallsetminus}\left\{  0\right\}  }\left(  e^{ist}%
-1-\frac{ist}{1+s^{2}}\right)  dM(s)\text{, }t\in\mathbb{R}\text{,}%
\]
where $a$, $\sigma\in\mathbb{R}$, with $\sigma\geq0$, and the measure $dM(s)$
satisfies%
\begin{equation}
\int\nolimits_{\mathbb{R\smallsetminus}\left\{  0\right\}  }\min\left(
1,s^{2}\right)  dM(s)<\infty. \label{dM(s)}%
\end{equation}
On the other hand, given a triple $\left(  a,\sigma,dM\right)  $\ with
$a\in\mathbb{R}$, $\sigma\geq0$, and $dM$ a measure on
$\mathbb{R\smallsetminus}\left\{  0\right\}  $ satisfying (\ref{dM(s)}), there
exists a unique infinitely divisible probability distribution $P$ such that
its L\'{e}vy characteristic is given by (\ref{char_function}).

Let $\digamma$ be a L\'{e}vy characteristic defined by (\ref{char_function}).
Then there exists a unique probability measure \textrm{P}$_{\digamma}$ on
$\left(  \mathcal{H}_{\infty}^{^{\ast}}\left(  \mathbb{R}\right)
,\mathcal{B}\right)  $ such that the `Fourier transform' of \textrm{P}%
$_{\digamma}$ satisfies%
\begin{equation}
\int\nolimits_{\mathcal{H}_{\infty}^{^{\ast}}\left(  \mathbb{R}\right)
}e^{i\left\langle W,f\right\rangle }\mathrm{dP}_{\digamma}\left(  W\right)
=\exp\left\{  \int\nolimits_{\mathbb{Q}_{p}^{N}}\digamma\left(  f\left(
x\right)  \right)  d^{N}x\right\}  \text{, }f\in\mathcal{H}_{\infty}\left(
\mathbb{R}\right)  , \label{Eq_the_zuniga}%
\end{equation}
cf. \cite[Theorem 5.2]{Zuniga-FAA-2017}, alternatively \cite[Theorem
11.6]{KKZuniga}.

We will say that a distribution $\boldsymbol{\Theta}\in\left(  \mathcal{H}%
_{\infty}\right)  ^{-1}$ is \textit{represented by a probability measure}
$\mathrm{P}$ on $\left(  \mathcal{H}_{\mathbb{R}}^{^{\ast}}\left(
\infty\right)  ,\mathcal{B}\right)  $ if
\begin{equation}
\left\langle \left\langle \boldsymbol{\Theta},\Psi\right\rangle \right\rangle
=%
{\textstyle\int\nolimits_{\mathcal{H}_{\mathbb{\infty}}^{^{\ast}}\left(
\mathbb{R}\right)  }}
\Psi\left(  W\right)  d\mathrm{P}\left(  W\right)  \text{ for any }\Psi
\in\left(  \mathcal{H}_{\infty}\right)  ^{1}. \label{Eq_definition}%
\end{equation}
We will denote this fact as $d\mathrm{P}=\boldsymbol{\Theta}d\mu$. In this
case $\boldsymbol{\Theta}$ may be regarded as the generalized Radon-Nikodym
derivative $\frac{d\mathrm{P}}{d\mu}$ of $\mathrm{P}$ with respect to $\mu$.

By using this result, Theorem \ref{Theorem1}-(iii), and assuming that
\begin{equation}
\digamma\left(  t\right)  =-H(it)-\frac{1}{2}t^{2}\text{, }t\in\mathbb{R}
\label{Levy_char}%
\end{equation}
is a L\'{e}vy characteristic, there exists a probability measure
$\mathrm{P}_{H}$ on $\left(  \mathcal{H}_{\mathbb{R}}^{^{\ast}}\left(
\infty\right)  ,\mathcal{B}\right)  $ such that%
\begin{equation}
{\LARGE T}\boldsymbol{\Phi}_{H}\left(  f\right)  =%
{\textstyle\int\nolimits_{\mathcal{H}_{\mathbb{\infty}}^{^{\ast}}\left(
\mathbb{R}\right)  }}
\exp\left(  i\left\langle W,f\right\rangle \right)  d\mathrm{P}_{H}\left(
W\right)  \text{, \ }f\in\mathcal{H}_{\infty}\left(  \mathbb{R}\right)
\text{.} \label{Eq_the_zuniga_2}%
\end{equation}

\begin{theorem}
\label{Theorem2A}Assume that $H$ is a holomorphic function at the origin
satisfying $H(0)=0$. Then $d\mathrm{P}_{H}=\boldsymbol{\Phi}_{H}d\mu$ if and
only if $\digamma\left(  t\right)  $ is a L\'{e}vy characteristic.
\end{theorem}

\begin{proof}
Assume that $\digamma\left(  t\right)  $ is a L\'{e}vy characteristic. By
(\ref{Eq_the_zuniga_2}), we have
\begin{equation}
{\LARGE T}\boldsymbol{\Phi}_{H}\left(  \lambda f\right)  =%
{\textstyle\int\nolimits_{\mathcal{H}_{\mathbb{\infty}}^{^{\ast}}\left(
\mathbb{R}\right)  }}
\exp\left(  \lambda i\left\langle W,f\right\rangle \right)  d\mathrm{P}%
_{H}\left(  W\right)  =\left\langle \left\langle \boldsymbol{\Phi}_{H}%
,\exp\left(  \lambda i\left\langle W,f\right\rangle \right)  \right\rangle
\right\rangle , \label{Eq_16AA}%
\end{equation}
for any $\lambda\in\mathbb{R}$.

In order to establish (\ref{Eq_definition}), it is sufficient to show that
(\ref{Eq_definition}) holds for $\Psi$ in a dense subspace of \ $\left(
L_{\mathbb{C}}^{2}\right)  $, we can choose the linear span of the exponential
functions of the form $\exp\alpha\left\langle W,f\right\rangle $ for
$\alpha\in\mathbb{C}$, $f\in\mathcal{H}_{\infty}\left(  \mathbb{R}\right)  $,
cf. \cite[Proposition 1.9]{Hida et al}. On the other hand, since
$\boldsymbol{\Phi}_{H}\in\mathcal{H}_{-l,-k.-1}(\mathbb{C})$ for some $l$,
$k\in\mathbb{N}$, and $\left(  L_{\mathbb{C}}^{2}\right)  $ is dense in
$\mathcal{H}_{-l,-k.-1}(\mathbb{C})$, it is sufficient to establish
(\ref{Eq_definition}) when $\boldsymbol{\Phi}_{H}\in\left(  L_{\mathbb{C}}%
^{2}\right)  $. \ Now the result follows from (\ref{Eq_16AA}) by using the
fact that
\[
\lambda\rightarrow{\LARGE T}\boldsymbol{\Phi}_{H}\left(  \lambda f\right)  =%
{\textstyle\int\nolimits_{\mathcal{H}_{\mathbb{\infty}}^{^{\ast}}\left(
\mathbb{R}\right)  }}
\exp\left(  \lambda i\left\langle W,f\right\rangle \right)  d\mathrm{\mu
}\left(  W\right)  \text{, }\lambda\in\mathbb{R}\text{,}%
\]
has an entire analytic extension, cf. \cite[Proposition 2.2]{Hida et al}.

Conversely, assume that $d\mathrm{P}_{H}=\boldsymbol{\Phi}_{H}d\mu$, then by
Theorem \ref{Theorem1}-(iii), we have%
\begin{gather}%
{\textstyle\int\nolimits_{\mathcal{H}_{\mathbb{\infty}}^{^{\ast}}\left(
\mathbb{R}\right)  }}
e^{i\left\langle W,f\right\rangle }d\mathrm{P}_{H}\left(  W\right)  =%
{\textstyle\int\nolimits_{\mathcal{H}_{\mathbb{\infty}}^{^{\ast}}\left(
\mathbb{R}\right)  }}
e^{i\left\langle W,f\right\rangle }\boldsymbol{\Phi}_{H}\left(  W\right)
d\mu\left(  W\right) \label{Eq_16BB}\\
=\left\langle \left\langle \boldsymbol{\Phi}_{H},e^{i\left\langle
\cdot,f\right\rangle }\right\rangle \right\rangle ={\LARGE T}\boldsymbol{\Phi
}_{H}\left(  f\right)  =\exp\left\{  \int\nolimits_{\mathbb{Q}_{p}^{N}%
}\digamma\left(  f\left(  x\right)  \right)  d^{N}x\right\}  \text{,
}\nonumber
\end{gather}
for $f\in\mathcal{H}_{\infty}\left(  \mathbb{R}\right)  $. We now take
$f\left(  x\right)  =t1_{\mathbb{Z}_{p}^{N}}\left(  x\right)  $, where
$t\in\mathbb{R}$ and $1_{\mathbb{Z}_{p}^{N}}$\ is the characteristic function
of $\mathbb{Z}_{p}^{N}$. By using that $H(0)=0$, we have
\begin{equation}
\exp\left\{  \int\nolimits_{\mathbb{Q}_{p}^{N}}\digamma\left(  f\left(
x\right)  \right)  d^{N}x\right\}  =\exp\digamma(t). \label{Eq_16CC}%
\end{equation}
Now, we consider the random variable:%
\[%
\begin{array}
[c]{cccc}%
\left\langle \cdot,1_{\mathbb{Z}_{p}^{N}}\right\rangle : & \left(
\mathcal{H}_{\infty}^{\ast}\left(  \mathbb{R}\right)  ,\mathcal{B}%
,\mathrm{P}_{H}\right)  & \rightarrow & \left(  \mathbb{R},\mathcal{B}%
(\mathbb{R})\right) \\
&  &  & \\
& W & \rightarrow & \left\langle W,1_{\mathbb{Z}_{p}^{N}}\right\rangle ,
\end{array}
\]
with probability distribution $\nu_{\left\langle \cdot,1_{\mathbb{Z}_{p}^{N}%
}\right\rangle }\left(  A\right)  =\mathrm{P}_{H}\left\{  W\in\mathcal{H}%
_{\infty}^{\ast}\left(  \mathbb{R}\right)  ;\left\langle W,1_{\mathbb{Z}%
_{p}^{N}}\right\rangle \in A\right\}  $, where $A$ is a Borel subset of
$\mathbb{R}$. Then, by (\ref{Eq_16BB})-(\ref{Eq_16CC}),
\begin{equation}%
{\textstyle\int\nolimits_{\mathcal{H}_{\mathbb{\infty}}^{^{\ast}}\left(
\mathbb{R}\right)  }}
e^{it\left\langle W,f\right\rangle }d\mathrm{P}_{H}\left(  W\right)
=\int\nolimits_{\mathbb{R}}e^{itz}d\nu_{\left\langle \cdot,1_{\mathbb{Z}%
_{p}^{N}}\right\rangle }\left(  z\right)  =\exp\digamma(t). \label{Eq_16DD}%
\end{equation}

\end{proof}

We call these measures \textit{generalized white noise measures}. The moments
of the measure $\mathrm{P}_{H}$ are the Schwinger functions $\left\{
\mathcal{S}_{n}^{\boldsymbol{\Phi}_{H}}\right\}  _{n\in\mathbb{N}}$.

Since $\left(  \mathcal{G}_{\alpha,m}f\right)  \left(  x\right)  :=G\left(
x;\alpha,m\right)  \ast f\left(  x\right)  $ gives rise to a continuous
mapping from $\mathcal{H}_{\mathbb{R}}(\infty)$ into itself, then, the
conjugate operator $\widetilde{\mathcal{G}}_{\alpha,m}:\mathcal{H}%
_{\mathbb{R}}^{^{\ast}}(\infty)\rightarrow\mathcal{H}_{\mathbb{R}}^{^{\ast}%
}(\infty)$ is a measurable mapping from $\left(  \mathcal{H}_{\mathbb{R}%
}^{^{\ast}}\left(  \infty\right)  ,\mathcal{B}\right)  $ into itself. For the
sake of simplicity, we use $\mathcal{G}$ instead of $\mathcal{G}_{\alpha,m}$
and $G$ instead of $G\left(  x;\alpha,m\right)  $. We set $\mathrm{P}_{H}^{G}$
to be the image probability measure of $\mathrm{P}_{H}$ under $\widetilde
{\mathcal{G}}$, i.e. $\mathrm{P}_{H}^{G}$ is the measure on $\left(
\mathcal{H}_{\mathbb{R}}^{^{\ast}}\left(  \infty\right)  ,\mathcal{B}\right)
$ defined by
\begin{equation}
\mathrm{P}_{H}^{G}\left(  A\right)  =\mathrm{P}_{H}\left(  \widetilde
{\mathcal{G}}^{-1}\left(  A\right)  \right)  \text{, for }A\in\mathcal{B}%
\text{.} \label{probability}%
\end{equation}

The Fourier transform of $\mathrm{P}_{H}^{G}$ is given by%
\begin{equation}%
{\textstyle\int\nolimits_{\mathcal{H}_{\mathbb{\infty}}^{^{\ast}}\left(
\mathbb{R}\right)  }}
e^{i\left\langle W,f\right\rangle }d\mathrm{P}_{H}^{G}\left(  W\right)
=\exp\left\{  \int\nolimits_{\mathbb{Q}_{p}^{N}}\digamma\left\{
\int\nolimits_{\mathbb{Q}_{p}^{N}}G\left(  x-y;\alpha,m\right)  f\left(
y\right)  d^{N}y\right\}  d^{N}x\right\}  \text{,} \label{Eq_16A}%
\end{equation}
for $f\in\mathcal{H}_{\mathbb{R}}\left(  \infty\right)  $, where $\digamma$ is
given as in (\ref{Levy_char}), cf. \cite[Proposition 6.2]{Zuniga-FAA-2017},
alternatively \cite[Proposition 11.12]{KKZuniga}. Finally, (\ref{Eq_16A}) is
also valid if we replace $G=$ $G\left(  x;\alpha,m\right)  $ by $G\left(
x;\alpha,\beta,m\right)  $.

\subsection{The free Euclidean Bose field}

An important difference between the real and $p$-adic Euclidean quantum field
theories comes from the `ellipticity' of the quadratic form $\mathfrak{q}%
_{N}\left(  \xi\right)  =\xi_{1}^{2}+\cdots+\xi_{N}^{2}$. In the real case
$\mathfrak{q}_{N}\left(  \xi\right)  $\ is elliptic for any $N\geq1$. In the
$p$-adic case, $\mathfrak{q}_{N}\left(  \xi\right)  $ is not elliptic for
$N\geq5$. In the case $N=4$, there is a unique elliptic quadratic form, up to
linear equivalence, which is $\xi_{1}^{2}-s\xi_{2}^{2}-p\xi_{3}^{2}+s\xi
_{4}^{2}$, where $s\in\mathbb{Z\smallsetminus}\left\{  0\right\}  $ is a
quadratic non-residue, i.e. $\left(  \frac{s}{p}\right)  =-1$.

\subsubsection{The Archimedean free covariance function}

The free covariance function $C(x-y;m):=C(x-y)$ is the solution of the Laplace
equation%
\[
\left(  -\Delta+m^{2}\right)  C(x-y)=\delta\left(  x-y\right)  ,
\]
where $\Delta=\sum_{i=1}^{N}\frac{\partial^{2}}{\partial x_{i}^{2}}$. As a
distribution from $\mathcal{S}^{\prime}(\mathbb{R}^{N})$, the free covariance
is given by
\[
C(x-y)=\frac{1}{\left(  2\pi\right)  ^{N}}\int\limits_{\mathbb{R}^{N}}%
\frac{\exp\left(  -ik\cdot\left(  x-y\right)  \right)  }{k^{2}+m^{2}}d^{N}k,
\]
where $k$, $x$, $y\in\mathbb{R}^{N}$, $d^{N}k$ is the Lebesgue measure of
$\mathbb{R}^{N}$, $k^{2}=k\cdot k$, and $k\cdot x=\sum_{i=1}^{N}k_{i}x_{i}$.
Notice that the quadratic form used in the definition of the Fourier transform
$k^{2}$ is the same as the one used in the propagator $\frac{1}{k^{2}+m^{2}}$,
this situation does not occur in the $p$-adic case. In particular the group of
symmetries of $C(x-y)$ is the $SO(N,\mathbb{R})$. The function $C(x-y)$ has
the following properties (see \cite[Proposition 7.2.1]{Glimm-Jaffe}.):

\noindent(i) $C(x-y)$ is positive and analytic for $x-y\neq0$;

\noindent(ii) $C(x-y)\leq\exp\left(  -m\left\Vert x-y\right\Vert \right)  $ as
$\left\Vert x-y\right\Vert \rightarrow\infty$;

\noindent(iii) for $N\geq3$ and $m\left\Vert x-y\right\Vert $ in a
neighborhood of zero,%
\[
C(x-y)\sim\left\Vert x-y\right\Vert ^{-N+2},
\]

\noindent(iv) for $N=2$ and $m\left\Vert x-y\right\Vert $ in a neighborhood of
zero,%
\[
C(x-y)\sim-\ln\left(  m\left\Vert x-y\right\Vert \right)  .
\]

\subsubsection{The Archimedean free Euclidean Bose field}

Take $H_{m}$ to be the Hilbert space defined as the closure of $\mathcal{S}%
(\mathbb{R}^{N})$ with respect to the norm $\left\Vert \cdot\right\Vert _{m}$
induced by the scalar product%
\[
\left(  f,g\right)  _{m}:=\int\nolimits_{\mathbb{R}^{N}}f\left(  x\right)
\left(  -\Delta+m^{2}\right)  ^{-1}g\left(  x\right)  d^{N}x=\left(  f,\left(
-\Delta+m^{2}\right)  ^{-1}g\right)  _{L^{2}(\mathbb{R}^{N})}.
\]
Then $\mathcal{S}(\mathbb{R}^{N})\hookrightarrow H_{m}\hookrightarrow
\mathcal{S}^{\prime}(\mathbb{R}^{N})$ form a Gel'fand triple. The probability
space $\left(  \mathcal{S}^{\prime}(\mathbb{R}^{N}),\mathcal{B},\nu\right)  $,
where $\nu$\ is the centered Gaussian measure on $\mathcal{B}$ (the $\sigma
$-algebra of cylinder sets) with covariance%
\[
\int\nolimits_{\mathcal{S}^{\prime}(\mathbb{R}^{N})}\left\langle
W,f\right\rangle \left\langle W,g\right\rangle d\nu\left(  W\right)  =\left(
f,\left(  -\Delta+m^{2}\right)  ^{-1}g\right)  _{L^{2}(\mathbb{R}^{N})},
\]
for $f$, $g$ $\in\mathcal{S}(\mathbb{R}^{N})$, jointly with the coordinate
process $W\rightarrow\left\langle W,f\right\rangle $, with \ fixed
$f\in\mathcal{S}(\mathbb{R}^{N})$, is called the free Euclidean Bose field of
mass $m$ in $N$ dimensions.

\subsubsection{\label{Section_C_p}The non-Archimedean free covariance
function}

The $p$-adic free covariance $C_{p}(x-y;m):=C_{p}(x-y)$ is the solution of the
pseudodifferential equation%
\[
\left(  \boldsymbol{L}_{\alpha}+m^{2}\right)  C(x-y)=\delta\left(  x-y\right)
,
\]
where $\boldsymbol{L}_{\alpha}$ is the pseudodifferential operator defined in
(\ref{ellipticoperaator}). As a distribution from $\mathcal{D}^{\prime
}(\mathbb{Q}_{p}^{N})$, the free covariance is given by%
\[
C_{p}(x-y)=\int\limits_{\mathbb{Q}_{p}^{N}}\frac{\chi_{p}\left(  -\xi
\cdot\left(  x-y\right)  \right)  }{\left\vert \mathfrak{l}(\xi)\right\vert
_{p}^{\alpha}+m^{2}}d^{N}\xi,
\]
where $k$, $x$, $y\in\mathbb{Q}_{p}^{N}$, $d^{N}\xi$ is the Haar measure of
$\mathbb{Q}_{p}^{N}$, $\mathfrak{l}(k)$ is an elliptic polynomial of degree
$d$, and $k\cdot x=\sum_{i=1}^{N}k_{i}x_{i}$.\ In this case $\mathfrak{l}%
(k)\neq k\cdot k$, and then the symmetries of $C_{p}(x-y)$ form a subgroup of
\ the $p$-adic orthogonal group attached to the quadratic form $k\cdot k$.
There are other possible propagators, for instance
\[
\frac{1}{\left(  \left\vert \mathfrak{l}(k)\right\vert _{p}+m^{2}\right)
^{\alpha}}\text{, }\alpha>0.
\]
For a discussion on the possible scalar propagators, in the $p$-adic setting,
the reader may consult \cite{Smirnov}.

The function $C_{p}(x-y)$ satisfies (see \cite[Proposition 4.1]%
{Zuniga-FAA-2017}, or \cite[Proposition 11.1]{KKZuniga}):

\noindent(i) $C_{p}(x-y)$ is positive and locally constant for $x-y\neq0$;

\noindent(ii) $C_{p}(x-y)\leq C\left\Vert x-y\right\Vert _{p}^{-\alpha d-N}$
as $\left\Vert x-y\right\Vert _{p}\rightarrow\infty$;

\noindent(iii) for $0<\alpha d<N$ and $\left\Vert x-y\right\Vert _{p}\leq1$,%
\[
C_{p}(x-y)\leq C\left\Vert x-y\right\Vert _{p}^{\alpha d-N};
\]

\noindent(iv) for $N=\alpha d$ and $\left\Vert x-y\right\Vert _{p}\leq1$,%
\[
C_{p}(x-y)\leq C_{0}-C_{1}\ln\left\Vert x-y\right\Vert _{p}.
\]

\subsubsection{The non-Archimedean free Euclidean Bose field}

Take $H_{m}$ to be the Hilbert space defined as the closure of $\mathcal{D}%
_{\mathbb{R}}(\mathbb{Q}_{p}^{N})$ with respect to the norm $\left\Vert
\cdot\right\Vert _{m}$ induced by the scalar product%
\[
\left(  f,g\right)  _{m}:=\int\nolimits_{\mathbb{Q}_{p}^{N}}\overline
{\widehat{f}\left(  \xi\right)  }\widehat{g}\left(  \xi\right)  \frac{d^{N}%
\xi}{\left\vert \mathfrak{l}(\xi)\right\vert _{p}^{\alpha}+m^{2}}=\left(
f,\left(  \boldsymbol{L}_{\alpha}+m^{2}\right)  ^{-1}g\right)  _{L_{\mathbb{R}%
}^{2}(\mathbb{Q}_{p}^{N})}.
\]
By using that
\[
C_{0}\left[  \xi\right]  _{p}^{\left\lfloor d\alpha\right\rfloor }%
\leq\left\vert \mathfrak{l}(\xi)\right\vert _{p}^{\alpha}+m^{2}\leq
C_{1}\left[  \xi\right]  _{p}^{\left\lceil d\alpha\right\rceil },
\]
where $\left\lceil t\right\rceil =\min\left\{  m\in\mathbb{Z};m\geq x\right\}
$ and $\left\lfloor t\right\rfloor =\max\left\{  m\in\mathbb{Z};m\leq
x\right\}  $, we have%
\[
\mathcal{H}_{-\left\lfloor d\alpha\right\rfloor }(\mathbb{R})\hookrightarrow
H_{m}\hookrightarrow\mathcal{H}_{-\left\lceil d\alpha\right\rceil }%
(\mathbb{R}).
\]

Then $\mathcal{H}_{\infty}(\mathbb{R})\hookrightarrow H_{m}\hookrightarrow
\mathcal{H}_{\infty}^{\ast}(\mathbb{R})$ from a Gel'fand triple. The
probability space $\left(  \mathcal{H}_{\infty}^{\ast}(\mathbb{R}%
),\mathcal{B},\nu_{d,\alpha}\right)  $, where $\nu_{d,\alpha}$\ is the
centered Gaussian measure on $\mathcal{B}$ (the $\sigma$-algebra of cylinder
sets) with covariance%
\[
\int\nolimits_{\mathcal{H}_{\infty}^{\ast}(\mathbb{R})}\left\langle
W,f\right\rangle \left\langle W,g\right\rangle d\nu_{d,\alpha}(W)=\left(
f,\left(  \boldsymbol{L}_{\alpha}+m^{2}\right)  ^{-1}g\right)  _{L_{\mathbb{R}%
}^{2}(\mathbb{Q}_{p}^{N})},
\]
for $f$, $g$ $\in\mathcal{H}_{\infty}(\mathbb{R})$, jointly with the
coordinate process $W\rightarrow\left\langle W,f\right\rangle $, with \ fixed
$f\in\mathcal{H}_{\infty}(\mathbb{R}))$, is called the non-Archimedean free
Euclidean Bose field of mass $m$ in $N$ dimensions.

If $N=4$ and $d=2$, then there is a unique elliptic quadratic form up to
linear equivalence. If $N\geq5$ and $\mathfrak{l}(\xi)$ is an elliptic
polynomial of degree $d$, then $\left\vert \mathfrak{l}(\xi)\right\vert
_{p}^{\frac{2}{d}}$ is a homogeneous function of degree $2$ that vanishes only
at the origin. We can use this function as the symbol for a pseudodifferential
operator, such operator is a $p$-adic analogue of $-\Delta$ in dimension $N$.

If we use the propagator $\frac{1}{\left(  \left\vert \mathfrak{l}%
(k)\right\vert _{p}+m^{2}\right)  ^{\alpha}}$ instead of $\frac{1}{\left\vert
\mathfrak{l}(k)\right\vert _{p}^{\alpha}+m^{2}}$, similar results are obtained
due to the fact that
\[
\text{and }C_{0}^{\prime}\left[  k\right]  _{p}^{\left\lfloor d\alpha
\right\rfloor }\leq\left(  \left\vert \mathfrak{l}(k)\right\vert _{p}%
+m^{2}\right)  ^{\alpha}\leq C_{1}^{\prime}\left[  k\right]  _{p}^{\left\lceil
d\alpha\right\rceil }.
\]
We prefer using propagator $\frac{1}{\left\vert \mathfrak{l}(k)\right\vert
_{p}^{\alpha}+m^{2}}$ because the corresponding `Laplace equation' has been
studied extensively in the literature. On the other hand, $\frac{\partial
u\left(  x,t\right)  }{\partial t}+\boldsymbol{L}_{\alpha}u\left(  x,t\right)
=0$, with $x\in\mathbb{Q}_{p}^{N}$, $t>0$, behaves like a `heat equation',
i.e. the semigroup associated to\ this equation is a Markov semigroup, see
\cite[Chapter 2]{Zuniga-LNM-2016}, which means that $-\boldsymbol{L}_{\alpha}$
can be considered \ as $p$-adic version of \ the Laplacian.

\subsection{\label{Sect_symmetries}Symmetries}

Given a polynomial $\mathfrak{a}\left(  \xi\right)  \in\mathbb{Q}_{p}\left[
\xi_{1},\cdots,\xi_{n}\right]  $ and $\Lambda\in GL_{N}\left(  \mathbb{Q}%
_{p}\right)  $, we say that $\Lambda$ \textit{preserves} $\mathfrak{a}$ if
$\mathfrak{a}\left(  \xi\right)  =\mathfrak{a}\left(  \Lambda\xi\right)  $,
for all $\xi\in\mathbb{Q}_{p}^{N}$. By simplicity, we use $\Lambda x$ to mean
$\left[  \Lambda_{ij}\right]  x^{T}$, $x=\left(  x_{1},\cdots,x_{N}\right)
\in\mathbb{Q}_{p}^{N}$, where we identify $\Lambda$ with the matrix $\left[
\Lambda_{ij}\right]  $.

Let $\mathfrak{q}_{N}\left(  \xi\right)  =\xi_{1}^{2}+\cdots+\xi_{N}^{2}$ be
the elliptic quadratic form used in the definition of the Fourier transform,
and let $\mathfrak{l}\left(  \xi\right)  $ be the elliptic polynomial that
appears in the symbol of the operator $\boldsymbol{L}_{\alpha}$. We define the
homogeneous Euclidean group of $\mathbb{Q}_{p}^{N}$ relative to $\mathfrak{q}%
\left(  \xi\right)  $\ and $\mathfrak{l}\left(  \xi\right)  $, denoted as
$E_{0}\left(  \mathbb{Q}_{p}^{N}\right)  :=E_{0}\left(  \mathbb{Q}_{p}%
^{N};\mathfrak{q},\mathfrak{l}\right)  $, as the subgroup of $GL_{N}\left(
\mathbb{Q}_{p}\right)  $ whose elements preserve $\mathfrak{q}\left(
\xi\right)  $\ and $\mathfrak{l}\left(  \xi\right)  $ simultaneously. Notice
that if $\mathbb{O}(\mathfrak{q}_{N})$ is the orthogonal group of
$\mathfrak{q}_{N}$, then $E_{0}\left(  \mathbb{Q}_{p}^{N}\right)  $ is a
subgroup of $\mathbb{O}(\mathfrak{q}_{N})$. We define the inhomogeneous
Euclidean group, denoted as $E\left(  \mathbb{Q}_{p}^{N}\right)
\allowbreak:=E\left(  \mathbb{Q}_{p}^{N};\mathfrak{q},\mathfrak{l}\right)  $,
to be the group of transformations of the form $\left(  a,\Lambda\right)
x=a+\Lambda x$, for $a,x\in\mathbb{Q}_{p}^{N}$, $\Lambda\in E_{0}\left(
\mathbb{Q}_{p}^{N}\right)  $.

In the real case $\mathfrak{q}_{N}=\mathfrak{l}\left(  \xi\right)  $ and thus
the homogeneous Euclidean group is $SO(N,\mathbb{R})$. In the $p$-adic case,
$E_{0}\left(  \mathbb{Q}_{p}^{N};\mathfrak{q},\mathfrak{l}\right)  $ \ is a
subgroup of $\mathbb{O}(\mathfrak{q}_{N})$, in addition, it is not a
straightforward matter to decide whether or not $E_{0}\left(  \mathbb{Q}%
_{p}^{N};\mathfrak{q},\mathfrak{l}\right)  $ is non trivial. For this reason,
we approach the Green kernels in a different way than do in \cite{GS1999},
which is based on \cite{Streater and Wightman}.

Notice that $\left(  a,\Lambda\right)  ^{-1}x=\Lambda^{-1}\left(  x-a\right)
$. Let $\left(  a,\Lambda\right)  $ be a transformation in $E\left(
\mathbb{Q}_{p}^{N}\right)  $, the action of $\left(  a,\Lambda\right)  $ on a
function $f\in\mathcal{H}_{\mathbb{\infty}}$ is defined by%
\[
\left(  \left(  a,\Lambda\right)  f\right)  \left(  x\right)  =f\left(
\left(  a,\Lambda\right)  ^{-1}x\right)  \text{, for\ }x\in\mathbb{Q}_{p}%
^{N},
\]
and on a functional $W\in\mathcal{H}_{\mathbb{\infty}}^{^{\ast}}$, by%
\[
\left\langle \left(  a,\Lambda\right)  W,f\right\rangle :=\left\langle
W,\left(  a,\Lambda\right)  ^{-1}f\right\rangle \text{, for }f\in
\mathcal{H}_{\mathbb{\infty}}\left(  \mathbb{R}\right)  .
\]
These definitions can be extended to elements of the spaces $\mathcal{H}%
_{\mathbb{\infty}}^{\otimes n}$ and $\mathcal{H}_{\mathbb{\infty}}%
^{\ast\otimes n}$, by taking
\[
\left(  a,\Lambda\right)  \left(  f_{1}\otimes\cdots\otimes f_{n}\right)
:=\left(  a,\Lambda\right)  ^{-1}f_{1}\otimes\cdots\otimes\left(
a,\Lambda\right)  ^{-1}f_{n}.
\]
In general, if $F:\mathcal{H}_{\mathbb{\infty}}^{\otimes n}\rightarrow
\mathcal{X}$ is linear $\mathcal{X}$-valued functional, where $\mathcal{X}$ is
a vector space, we define%
\[
\left(  \left(  a,\Lambda\right)  F\right)  \left(  f_{1}\otimes\cdots\otimes
f_{n}\right)  =F\left(  \left(  a,\Lambda\right)  \left(  f_{1}\otimes
\cdots\otimes f_{n}\right)  \right)  ,
\]
and we say that $F$ is \textit{Euclidean invariant} if and only if $\left(
a,\Lambda\right)  F=F$ for any $\left(  a,\Lambda\right)  \in E\left(
\mathbb{Q}_{p}^{N}\right)  $.

\begin{definition}
We call a distribution $\boldsymbol{\Phi}=\sum_{n=0}^{\infty}\left\langle
\Phi_{n},:\cdot^{\otimes n}:\right\rangle \in\left(  \mathcal{H}%
_{\mathbb{\infty}}\right)  ^{-1}$, with $\Phi_{n}\in\mathcal{H}%
_{\mathbb{\infty}}^{^{\ast\widehat{\otimes}n}}$, Euclidean invariant if and
only if the functional $\left\langle \Phi_{n},\cdot\right\rangle $ is
Euclidean invariant for any $n\in\mathbb{N}$.
\end{definition}

It follows from this definition that $\boldsymbol{\Phi}\in\left(
\mathcal{H}_{\mathbb{\infty}}\right)  ^{-1}$ is Euclidean invariant if and
only if ${\LARGE S}\boldsymbol{\Phi}$\ and ${\LARGE T}\boldsymbol{\Phi}$ are
Euclidean invariant.

\section{Schwinger functions and convoluted white noise}

We set $G:=$ $G\left(  x;m,\alpha\right)  $ for the Green function
(\ref{GreenFunc}). For $\boldsymbol{\Phi}\in\left(  \mathcal{H}%
_{\mathbb{\infty}}\right)  ^{-1}$, we define $\boldsymbol{\Phi}^{G}$\ as%
\begin{equation}
\left(  {\LARGE T}\boldsymbol{\Phi}^{G}\right)  \left(  g\right)  =\left(
{\LARGE T}\boldsymbol{\Phi}\right)  \left(  G\ast g\right)  \text{, \ }%
g\in\mathcal{U}\text{,}\label{Eq_15}%
\end{equation}
where $\mathcal{U}$ is an open neighborhood of zero. Since $\mathcal{G}%
:\mathcal{H}_{\mathbb{\infty}}\left(  \mathbb{R}\right)  \rightarrow
\mathcal{H}_{\mathbb{\infty}}\left(  \mathbb{R}\right)  $, see
(\ref{Mapping_G}), is linear and continuous, cf. \cite[Corollary
11.3]{KKZuniga}, by the characterization theorem, cf. \cite[Theorem 3]{KLS96},
or Section \ref{Section_Characterization}, $\boldsymbol{\Phi}^{G}$ is a
well-defined and unique element of $\left(  \mathcal{H}_{\mathbb{\infty}%
}\right)  ^{-1}$.

\begin{remark}
\label{Nota5}By using that $\left\langle \delta_{x},G\ast f\right\rangle
=\left\langle G\ast\delta_{x},f\right\rangle $ for any $f\in\mathcal{H}%
_{\mathbb{\infty}}$, we have that the white-noise process introduced in
Section \ref{Section_white_noise_process} satisfies
\[
\left\langle \left\langle G\ast\boldsymbol{\Phi}\left(  x\right)
,\Psi\right\rangle \right\rangle =\left\langle \left\langle \boldsymbol{\Phi
}\left(  x\right)  ,G\ast\Psi\right\rangle \right\rangle ,
\]
because $\left\langle \left\langle \boldsymbol{\Phi}\left(  x\right)
,G\ast\Psi\right\rangle \right\rangle =\left\langle \delta_{x},G\ast\psi
_{1}\right\rangle $, where $\Psi=\sum_{n=0}^{\infty}\left\langle
:\cdot^{\otimes n}:,\psi_{n}\right\rangle $, $\psi_{n}\in\mathcal{H}_{\infty
}^{\widehat{\otimes}n}$.
\end{remark}

We denote by $\left\{  S_{n}^{H,G}\right\}  _{n\in N}$ the Schwinger functions
attached to $\boldsymbol{\Phi}_{H}^{G}$.

\begin{theorem}
\label{Theorem2}With $H$ and $\boldsymbol{\Phi}_{H}$ as in Theorem
\ref{Theorem1}, then distribution $\boldsymbol{\Phi}_{H}^{G}\in\left(
\mathcal{H}_{\infty}\right)  ^{-1}$ is Euclidean invariant and is given by%
\begin{equation}
\boldsymbol{\Phi}_{H}^{G}=\exp^{\lozenge}\left(  -\int\nolimits_{\mathbb{Q}%
_{p}^{N}}H^{\lozenge}\left(  G\ast\boldsymbol{\Phi}\left(  x\right)  \right)
d^{N}x+\frac{1}{2}\left\langle \left(  \mathcal{G}^{\otimes2}-1\right)
Tr,:\cdot^{\otimes2}:\right\rangle \right)  , \label{Eq_16}%
\end{equation}
where $Tr\in\left(  \mathcal{H}_{\mathbb{\infty}}^{^{\ast}}\left(
\mathbb{Q}_{p}^{N},\mathbb{C}\right)  \right)  ^{\widehat{\otimes}2}$ denotes
the trace kernel defined by $\left\langle Tr,f\otimes g\right\rangle
=\left\langle f,g\right\rangle _{0}$, $f$, $g\in\mathcal{H}_{\mathbb{\infty}%
}\left(  \mathbb{Q}_{p}^{N},\mathbb{R}\right)  $. The Schwinger functions
$\left\{  S_{n}^{H,G}\right\}  _{n\in N}$ satisfy the conditions (OS1) and
(OS4) given in Lemma \ref{Lemma1}, and
\begin{equation}
\text{(Euclidean invariance) \ \ \ }S_{n}^{H,G}\left(  \left(  a,\Lambda
\right)  f\right)  =S_{n}^{H,G}\left(  f\right)  \text{, }f\in\left(
\mathcal{H}_{\mathbb{\infty}}\left(  \mathbb{C}\right)  \right)  ^{\otimes n},
\tag{OS2}%
\end{equation}
for any $\left(  a,\Lambda\right)  \in E\left(  \mathbb{Q}_{p}^{N}\right)  $.
\end{theorem}

\begin{proof}
By definition (\ref{Eq_15}) and Theorem \ref{Theorem1}-(iii), we have%
\begin{equation}
\left(  {\LARGE T}\boldsymbol{\Phi}_{H}^{G}\right)  \left(  g\right)
=\exp\left(  -\int\nolimits_{\mathbb{Q}_{p}^{N}}H(iG\ast g\left(  x\right)
)+\frac{1}{2}\left(  G\ast g\left(  x\right)  \right)  ^{2}\text{ }%
d^{N}x\right)  . \label{Eq_17}%
\end{equation}
On the other hand, by taking the ${\LARGE T}$-transform in (\ref{Eq_16}) and
using (\ref{Eq_13}) and Remarks \ref{Nota_tensor_2}-\ref{Nota5}, we obtain%
\begin{multline*}
\left(  {\LARGE T}\boldsymbol{\Phi}_{H}^{G}\right)  \left(  g\right)
=\exp\left(  \frac{-1}{2}\left\Vert g\right\Vert _{0}^{2}\right)  \times\\
\exp\left\{  -{\LARGE S}\left(  \int\nolimits_{\mathbb{Q}_{p}^{N}}H^{\lozenge
}\left(  G\ast\boldsymbol{\Phi}\left(  x\right)  \right)  d^{N}x\right)
\left(  ig\right)  -\frac{1}{2}{\LARGE S}\left(  \left\langle \left(
\mathcal{G}^{\otimes2}\mathcal{-}1\right)  Tr,:\cdot^{\otimes2}:\right\rangle
\right)  \left(  ig\right)  \right\}
\end{multline*}%
\begin{multline*}
=\exp\left(  \frac{-1}{2}\left\Vert g\right\Vert _{0}^{2}\right)  \times\\
\exp\left\{  -\int\nolimits_{\mathbb{Q}_{p}^{N}}H\left(  iG\ast g\left(
x\right)  \right)  d^{N}x\right\}  \exp\left\{  \frac{1}{2}{\LARGE S}\left(
\left\langle \left(  \mathcal{G}^{\otimes2}\mathcal{-}1\right)  Tr,:\cdot
^{\otimes2}:\right\rangle \right)  \left(  ig\right)  \right\} \\
=\exp\left(  \frac{-1}{2}\left\Vert g\right\Vert _{0}^{2}\right)  \exp\left\{
-\int\nolimits_{\mathbb{Q}_{p}^{N}}H\left(  iG\ast g\left(  x\right)  \right)
d^{N}x-\frac{1}{2}\left\langle \left(  \mathcal{G}^{\otimes2}\mathcal{-}%
1\right)  Tr,g\otimes g\right\rangle \right\} \\
=\exp\left(  \frac{-1}{2}\left\Vert g\right\Vert _{0}^{2}\right)  \exp\left\{
-\int\nolimits_{\mathbb{Q}_{p}^{N}}H\left(  iG\ast g\left(  x\right)  \right)
d^{N}x-\frac{1}{2}\left\langle Tr,G\ast g\otimes G\ast g-g\otimes
g\right\rangle \right\}
\end{multline*}%
\begin{equation}
=\exp\left\{  -\int\nolimits_{\mathbb{Q}_{p}^{N}}H\left(  iG\ast g\left(
x\right)  \right)  d^{N}x-\frac{1}{2}\left\langle G\ast g\otimes G\ast
g\right\rangle _{0}\right\}  \label{Eq_18AA}%
\end{equation}
Formula (\ref{Eq_16}) follows from (\ref{Eq_17})-(\ref{Eq_18AA}). Since
$\mathbb{E}_{\mu}(\boldsymbol{\Phi}_{H}^{G})=1$, conditions (OS1) and (OS4)
follow from Lemma \ref{Lemma1}, and condition (OS2) follows from Lemma
\ref{Lemma0} by using the Euclidean invariance of $\boldsymbol{\Phi}_{H}^{G}$.
\end{proof}

\begin{remark}
(i) Set $\mathcal{G}_{_{\frac{1}{2}}}:=\mathcal{G}_{\alpha,\frac{1}{2}%
,m}=\left(  \boldsymbol{L}_{\alpha}+m^{2}\right)  ^{-\frac{1}{2}}$, and
$\mathcal{G}_{_{\frac{1}{2}}}\left(  f\right)  :=G_{\frac{1}{2}}\ast f$ for
$f\in\mathcal{H}_{\infty}(\mathbb{R})$. By taking $H\equiv0$, we obtain the
free Euclidean field. Indeed, $f\rightarrow\exp\left\{  -\frac{1}%
{2}\left\langle G_{\frac{1}{2}}\ast f,G_{\frac{1}{2}}\ast f\right\rangle
_{0}\right\}  $ defines a characteristic functional. Let denote by
$\nu_{G_{\frac{1}{2}}}$ the probability measure on $\left(  \mathcal{H}%
_{\infty}^{\ast}(\mathbb{R}),\mathcal{B}\right)  $ provided by the
Bochner-Minlos theorem. Then%
\begin{align*}
\left(  {\LARGE T}\boldsymbol{\Phi}_{0}^{G_{\frac{1}{2}}}\right)  \left(
g\right)   &  =\exp\left\{  -\frac{1}{2}\left\langle G_{\frac{1}{2}}\ast
g,G_{\frac{1}{2}}\ast g\right\rangle _{0}\right\}  \\
&  =\left\langle \left\langle \boldsymbol{\Phi}_{0}^{G_{\frac{1}{2}}},\exp
i\left\langle \cdot,g\right\rangle \right\rangle \right\rangle =\int
\nolimits_{\mathcal{H}_{\infty}^{\ast}(\mathbb{R})}\exp i\left\langle
W,g\right\rangle d\nu_{G_{\frac{1}{2}}}(W).
\end{align*}

(ii )Assuming that $\digamma\left(  t\right)  $, see (\ref{Levy_char}), is a
L\'{e}vy characteristic, Theorem \ref{Theorem2} implies that the probability
measure $\mathrm{P}_{H}^{G}$, see (\ref{probability}), admits
$\boldsymbol{\Phi}_{H}^{G}$ as a generalized density with respect to to white
noise measure $\mu$, i.e. $\mathrm{P}_{H}^{G}=\boldsymbol{\Phi}_{H}^{G}\mu$.
Indeed, by (\ref{Eq_16A}) and (\ref{Eq_18AA}), we have
\begin{multline*}
\int\nolimits_{\mathcal{H}_{\mathbb{\infty}}^{^{\ast}}\left(  \mathbb{R}%
\right)  }e^{i\left\langle W,f\right\rangle }d\mathrm{P}_{H}^{G}\left(
W\right)  =\exp\left\{  \int\nolimits_{\mathbb{Q}_{p}^{N}}\digamma\left(
G\left(  x;\alpha,m\right)  \ast f\left(  x\right)  \right)  yd^{N}x\right\}
\\
=\exp\left\{  -\int\nolimits_{\mathbb{Q}_{p}^{N}}H\left(  iG\ast f\left(
x\right)  \right)  d^{N}x-\frac{1}{2}\left\langle G\ast f,G\ast f\right\rangle
_{0}\right\}  =\left(  {\LARGE T}\boldsymbol{\Phi}_{H}^{G}\right)  \left(
f\right) \\
=\left\langle \left\langle \boldsymbol{\Phi}_{H}^{G},\exp i\left\langle
\cdot,f\right\rangle \right\rangle \right\rangle .
\end{multline*}

\end{remark}

\subsection{Truncated Schwinger functions and the cluster property}

We denote by $P^{(n)}$ the collection of all partitions $I$ of $\left\{
1,\ldots,n\right\}  $ into disjoint subsets.

\begin{definition}
Let $\left\{  S_{n}^{H,G}\right\}  _{n\in\mathbb{N}}$ be a sequence of
Schwinger functions, with $S_{0}^{H,G}=1$, and $S_{n}^{H,G}\in\mathcal{H}%
_{\infty}^{\ast}\left(  \mathbb{Q}_{p}^{Nn},\mathbb{C}\right)  $ for $n\geq1$.
The truncated Schwinger functions $\left\{  S_{n,T}^{H,G}\right\}
_{n\in\mathbb{N}}$ are defined recursively by the formula%
\[
S_{n}^{H,G}\left(  f_{1}\otimes\cdots\otimes f_{n}\right)  =\sum\limits_{I\in
P^{(n)}}\prod\limits_{\left\{  j_{1},\ldots,j_{l}\right\}  }S_{l,T}%
^{H,G}\left(  f_{j_{1}}\otimes\cdots\otimes f_{j_{l}}\right)  ,
\]
for $n\geq1$. Here for $\left\{  j_{1},\ldots,j_{l}\right\}  \in I$ we assume
that $j_{1}<\ldots<j_{l}$.
\end{definition}

\begin{remark}
By the kernel theorem, the sequence $\left\{  S_{n}^{H,G}\right\}
_{n\in\mathbb{N}}$ uniquely determines the sequence $\left\{  S_{n,T}%
^{H,G}\right\}  _{n\in\mathbb{N}}$ and vice versa. \ All the $S_{n}^{H,G}$ are
Euclidean (translation) invariant if and only if all the $S_{n,T}^{H,G}$ are
Euclidean (translation) invariant. The same equivalence holds for
`temperedness' ( i.e. membership to $\left(  \mathcal{H}_{\infty}\right)
^{-1}$).
\end{remark}

\begin{definition}
Let $a\in\mathbb{Q}_{p}^{N}$, $a\neq0$, and $\lambda\in\mathbb{Q}_{p}$. Let
$T_{a\lambda}$ denote the representation of the translation by $a\lambda$ on
$\mathcal{H}_{\infty}\left(  \mathbb{Q}_{p}^{Nn},\mathbb{R}\right)  $. Take
$n$, $m\geq1$, $f_{1},\cdots,f_{n}\in\mathcal{H}_{\infty}\left(
\mathbb{Q}_{p}^{N},\mathbb{R}\right)  $.

(OS5)(\textbf{Cluster property}) A sequence of Schwinger functions $\left\{
S_{n}^{H,G}\right\}  _{n\in\mathbb{N}}$ has the cluster property if for all
$n$, $m\geq1$, it verifies that
\begin{align}
&  \lim_{\left\vert \lambda\right\vert _{p}\rightarrow\infty}\left\{
S_{m+n}^{H,G}\left(  f_{1}\otimes\cdots\otimes f_{m}\otimes T_{a\lambda
}\left(  f_{m+1}\otimes\cdots\otimes f_{m+n}\right)  \right)  \right\}
\label{Cluster_1}\\
&  =S_{m}^{H,G}\left(  f_{1}\otimes\cdots\otimes f_{m}\right)  S_{n}%
^{H,G}\left(  f_{m+1}\otimes\cdots\otimes f_{m+n}\right)  .\nonumber
\end{align}

(\textbf{Cluster property of truncated Schwinger functions}) A sequence of
truncated Schwinger functions $\left\{  S_{n,T}^{H,G}\right\}  _{n\in
\mathbb{N}}$ has the cluster property, if for all $n$, $m\geq1$, it verifies
that%
\begin{equation}
\lim_{\left\vert \lambda\right\vert _{p}\rightarrow\infty}S_{m+n,T}%
^{H,G}\left(  f_{1}\otimes\cdots\otimes f_{m}\otimes T_{a\lambda}\left(
f_{m+1}\otimes\cdots\otimes f_{m+n}\right)  \right)  =0. \label{Cluster_2}%
\end{equation}

\end{definition}

\begin{remark}
\label{Nota_Cluster}In the Archimedean case, it is possible to replace
$\lim_{\lambda\rightarrow\infty}\left(  \cdot\right)  $ in (\ref{Cluster_1})
and (\ref{Cluster_2}) by $\lim_{\lambda\rightarrow\infty}\left\vert
\lambda\right\vert ^{m}\left(  \cdot\right)  $ for arbitrary $m$, cf.
\cite[Remark 4.4]{Albeverio-et-al-3}. This is possible because Schwartz
functions decay at infinity faster than any polynomial function. This is not
possible in the $p$-adic case, because the elements of our `$p$-adic Schwartz
space $\mathcal{H}_{\infty}\left(  \mathbb{Q}_{p}^{N},\mathbb{R}\right)  $'
only have a polynomial decay at infinity. For instance, consider \ the
one-dimensional $p$-adic heat kernel \ $Z(x;t)=\mathcal{F}_{\xi\rightarrow
x}^{-1}\left(  e^{-t\left\vert \xi\right\vert _{p}^{\alpha}}\right)  $, \ for
$t>0$, and $\alpha>0$, which is an element of $\mathcal{H}_{\infty}\left(
\mathbb{Q}_{p},\mathbb{R}\right)  $. The Fourier transform $e^{-t\left\vert
\xi\right\vert _{p}^{\alpha}}$\ of \ $Z(x;t)$ decays faster that any
polynomial function in $\left\vert \xi\right\vert _{p}$. However, $Z(x;t)$ has
only a polynomial decay at infinity, \ more precisely,%
\[
Z(x;t)\leq C\frac{t}{\left(  t^{\frac{1}{\alpha}}+\left\vert x\right\vert
_{p}\right)  ^{\alpha+1}}\text{, }t>0\text{, }x\in\mathbb{Q}_{p},
\]
cf. \cite[Lemma 4.1]{Koch}.
\end{remark}

\begin{lemma}
\label{Lemma_4}Let $H(z)=%
{\textstyle\sum\nolimits_{n=0}^{\infty}}
H_{n}z^{n}$, $z\in U\subset\mathbb{C}$, and $G$ as in Theorem \ref{Theorem2},
and $f_{1},\cdots,f_{n}\in\mathcal{H}_{\infty}\left(  \mathbb{Q}_{p}%
^{N},\mathbb{R}\right)  $. Assume that $\digamma\left(  t\right)
=-H(it)-\frac{1}{2}t^{2}$, $t\in\mathbb{R}$ is a L\'{e}vy characteristic, then
the truncated Schwinger functions are given by%
\begin{equation}
S_{n,T}^{H,G}\left(  f_{1}\otimes\cdots\otimes f_{n}\right)  =\left\{
\begin{array}
[c]{lll}%
-H_{n}\int\nolimits_{\mathbb{Q}_{p}^{N}}\prod\limits_{i=1}^{n}G\ast
f_{i}\left(  x\right)  d^{N}x & \text{for } & n\geq2\\
&  & \\
(-H_{2}+1)\int\nolimits_{\mathbb{Q}_{p}^{N}}G\ast f_{1}\left(  x\right)  G\ast
f_{2}\left(  x\right)  d^{N}x & \text{for} & n=2.
\end{array}
\right.  \label{Eq_20}%
\end{equation}

\end{lemma}

\begin{proof}
The result follows from the formula for the Schwinger functions given in
Theorem 7.7 in \cite{Zuniga-FAA-2017}, and the uniqueness of the truncated
Schwinger functions. The coefficients in front of the integrals in
(\ref{Eq_20}) are the $n$-th derivatives of the L\'{e}vy characteristic
divided by $i^{n}$. For the general $H$ as in Theorem \ref{Theorem2} these
coefficients are the $n$-th derivatives of \ $-\left(  H(iz)+\frac{1}{2}%
z^{2}\right)  $, $z\in U$.
\end{proof}

\begin{lemma}
\label{Lemma_5}Assume that $\alpha d>N$. Let $\boldsymbol{\Phi}$, $H$, $G$ as
in Theorem \ref{Theorem2}. Then the sequence of truncated Schwinger functions
$\left\{  S_{n,T}^{H,G}\right\}  _{n\in\mathbb{N}}$ has the cluster property.
\end{lemma}

\begin{proof}
Fix $a\in\mathbb{Q}_{p}^{N}$ and take $\lambda\in\mathbb{Q}_{p}$, $m$,
$n\geq1$, $f_{1},\cdots,f_{m+n}\in\mathcal{H}_{\infty}\left(  \mathbb{Q}%
_{p}^{N},\mathbb{R}\right)  $. By Lemma \ref{Lemma_4}, we have
\begin{align*}
&  \left\vert S_{n,T}^{H,G}\left(  f_{1}\otimes\cdots\otimes f_{n}\right)
\otimes T_{a\lambda}\left(  f_{m+1}\otimes\cdots\otimes f_{m+n}\right)
\right\vert \\
&  =\left\vert -H_{m+n}\right\vert \left\vert \int\nolimits_{\mathbb{Q}%
_{p}^{N}}\prod\limits_{i=1}^{m}\left(  G\ast f_{i}\right)  \left(  x\right)
\text{ }\prod\limits_{i=m+1}^{m+n}T_{a\lambda}\left(  G\ast f_{i}\right)
\left(  x\right)  \text{ }\right\vert
\end{align*}
We now use that $G\ast f_{i}\in\mathcal{H}_{\infty}\left(  \mathbb{Q}_{p}%
^{N},\mathbb{R}\right)  $ and that $\mathcal{H}_{\infty}\left(  \mathbb{Q}%
_{p}^{N},\mathbb{R}\right)  \subset\mathcal{C}_{0}\left(  \mathbb{Q}_{p}%
^{N},\mathbb{R}\right)  $ to get%
\begin{align*}
&  \left\vert S_{n,T}^{H,G}\left(  f_{1}\otimes\cdots\otimes f_{n}\right)
\otimes T_{a\lambda}\left(  f_{m+1}\otimes\cdots\otimes f_{m+n}\right)
\right\vert \\
&  \leq\left\vert H_{m+n}\right\vert \prod\limits_{i=1}^{m}\left\Vert G\ast
f_{i}\right\Vert _{L^{\infty}}\text{ }\prod\limits_{i=m+1}^{m+n-1}\left\Vert
T_{a\lambda}\left(  G\ast f_{i}\right)  \right\Vert _{L^{\infty}}%
\int\nolimits_{\mathbb{Q}_{p}^{N}}\left\vert T_{a\lambda}\left(  G\ast
f_{m+n}\right)  \left(  x\right)  \right\vert d^{N}x.
\end{align*}
Now, the announced result follows from the following fact:

\textbf{Claim.} If $\alpha d>N$, for any $f\in\mathcal{H}_{\infty}\left(
\mathbb{Q}_{p}^{N},\mathbb{R}\right)  $, it verifies that%
\[
\lim_{\left\vert \lambda\right\vert _{p}\rightarrow\infty}\int
\nolimits_{\mathbb{Q}_{p}^{N}}G\left(  x-\lambda a-y\right)  \left\vert
f_{m+n}\left(  y\right)  \right\vert \text{ }d^{N}y=0.
\]
Since $\alpha d>N$, by the Riemann-Lebesgue theorem, $G\in\mathcal{C}%
_{0}\left(  \mathbb{Q}_{p}^{N},\mathbb{R}\right)  $, and consequently
$G\left(  x-\lambda a-y\right)  \left\vert f_{m+n}\left(  y\right)
\right\vert \leq\left\Vert G\right\Vert _{L^{\infty}}\left\vert f_{m+n}\left(
y\right)  \right\vert \in L_{\mathbb{R}}^{1}\left(  \mathbb{Q}_{p}^{N}\right)
$. Now the Claim follows by applying the dominated convergence theorem.
\end{proof}

\begin{theorem}
\label{Theorem3}With $H$, $G$ and $\boldsymbol{\Phi}_{H}^{G}\in\left(
\mathcal{H}_{\infty}\right)  ^{-1}$as in Theorem \ref{Theorem2}. If $\alpha
d>N$, then the sequence of Schwinger functions $\left\{  S_{n}^{H,G}\right\}
_{n\in\mathbb{N}}$ has the cluster property (OS5).
\end{theorem}

\begin{proof}
In \cite[Theorem 4.5]{Albeverio-et-al-3} was established that the cluster
property and the truncated cluster property are equivalent. By using this
result, the announced result follows from Lemma \ref{Lemma_5}.
\end{proof}

\begin{remark}
\label{Nota_Theorem_3}The class of Schwinger functions $\left\{  S_{n}%
^{H,G}\right\}  _{n\in\mathbb{N}}$ corresponding to a distribution
$\boldsymbol{\Phi}_{H}^{G}\in\left(  \mathcal{H}_{\infty}\right)  ^{-1}$as in
Theorem \ref{Theorem2} differs of the class of Schwinger functions
corresponding to the convoluted \ generalized white noise introduced in
\cite{Zuniga-FAA-2017}. In \ order to explain the differences, let us compare
the properties of the Levy characteristic used in \ \cite{Zuniga-FAA-2017}
with the properties of the function $H$ used in this article, where
$\digamma\left(  t\right)  =-H(it)-\frac{1}{2}t^{2}$, $t\in U\subset
\mathbb{R}$. We require only that function $H$ be holomorphic at zero and
$H(0)=0$, as in \cite{GS1999}. This only impose a restriction in choosing the
coefficients in front of the integrals corresponding to the $n$-th truncated
\ Schwinger function, see (\ref{Eq_20}). On the other hand in
\cite{Zuniga-FAA-2017}, the author requires the condition that the measure $M$
\ has finite moments of all orders. This implies that $\digamma$ belongs to
$C^{\infty}(\mathbb{R})$, but $\digamma$ does not have to have a holomorphic
extension. Furthermore, since $\exp s\digamma\left(  t\right)  $ is positive
definite for any $s>0$, cf. \cite[Proposition 5.5]{Zuniga-FAA-2017}, and by
using $\digamma\left(  0\right)  =0$ and a result due Schoenberg, cf.
\cite[Theorem 7.8]{Berg-Gunnar}, we have $-\digamma\left(  t\right)
:\mathbb{R}\rightarrow\mathbb{C}$ is a negative definite analytic function.
Since $\left\vert -\digamma\left(  t\right)  \right\vert \leq C\left\vert
t\right\vert ^{2}$ for any $\left\vert t\right\vert \geq1$, \cite[Corollary
7.16]{Berg-Gunnar}, we conclude that $-\digamma\left(  t\right)  $ is a
polynomial of the degree at most $2$, \ and then $H_{n}=0$ for $n\geq3$.
\end{remark}

\end{document}